\documentclass[aoas,preprint]{imsart}

\RequirePackage[OT1]{fontenc}
\RequirePackage{amsthm,amsmath}
\RequirePackage{natbib}
\RequirePackage[colorlinks,citecolor=blue,urlcolor=blue]{hyperref}
\usepackage{bbold}
\usepackage{tikz}
\usepackage{pgfplots}
\usepackage[ruled,vlined]{algorithm2e}

\usepackage{subcaption}
\usepackage{float}
\usepackage{comment}
\usepackage{adjustbox}
\usepackage{thmtools}

\usepackage[format=plain,labelfont=bf,textfont=it, font=small]{caption}

\usepackage{graphicx}
\usepackage{amsfonts}
\usepackage{wrapfig}

\newcommand{\Prob}{\ensuremath{{\mathbb P}}} 

\newcommand{\R}{\ensuremath{{\mathbb R}}}
\newcommand{\E}{\ensuremath{{\mathbb E}}}

\def\nottilde#1{\overset{\ifx#1f\hspace{.5ex}\fi\lower0.8ex\hbox{\tiny$\nsim$}}{#1}}

\usepackage{lscape}

\newcommand{\Ind}{\ensuremath{{\mathds{1}}}} 
\usepackage{dsfont}
\usepackage{caption}
\usepackage{float}
\usepackage{amsmath,amssymb}
\newcommand{\A}{\ensuremath{{\mathcal A}}}

\newcommand{\distas}[1]{\mathbin{\overset{#1}{\kern\z\sim}}}%
\newsavebox{\mybox}\newsavebox{\mysim}
\newcommand{\distras}[1]{%
  \savebox{\mybox}{\hbox{\kern3pt$\scriptstyle#1$\kern3pt}}%
  \savebox{\mysim}{\hbox{$\sim$}}%
  \mathbin{\overset{#1}{\kern\z@\resizebox{\wd\mybox}{\ht\mysim}{$\sim$}}}%
}

\usepackage{varioref}

\startlocaldefs
\numberwithin{equation}{section}
\theoremstyle{plain}

\newtheorem{definition}{Definition}[section]

\newtheorem{example}[definition]{Example}
\endlocaldefs

\begin{document}


\begin{abstract}
 We develop a fully non-parametric, easy-to-use, and powerful test for the missing completely at random (MCAR) assumption on the missingness mechanism of a dataset. The test compares distributions of different missing patterns on random projections in the variable space of the data. The distributional differences are measured with the Kullback-Leibler Divergence, using probability Random Forests \citep{probabilityforests}. We thus refer to it as ``Projected Kullback-Leibler MCAR'' (PKLM) test. The use of random projections makes it applicable  even if very few or no fully observed observations are available or if the number of dimensions is large. An efficient permutation approach guarantees the level for any finite sample size, resolving a major shortcoming of most other available tests. Moreover, the test can be used on both discrete and continuous data. We show empirically on a range of simulated data distributions and real datasets that our test has consistently high power and is able to avoid inflated type-I errors. Finally, we provide an \textsf{R}-package \texttt{PKLMtest} with an implementation of our test.
\end{abstract}

\begin{frontmatter}
\title{PKLM: A flexible MCAR test using Classification } 
\runtitle{PKLM}
\thankstext{T1}{Authors with equal contribution.}
\thankstext{T2}{We are grateful to Jun Li for providing us with parts of their code.}

\begin{aug}
\author{\fnms{Meta-Lina} \snm{Spohn}*},
\author{\fnms{Jeffrey} \snm{Näf}*}, %
\author{\fnms{Loris} \snm{Michel}},
\ead[label=u1,url]{http://www.foo.com}
\and
\author{\fnms{Nicolai} \snm{Meinshausen}}


\affiliation{ETH Zürich}

\address{Seminar for Statistics\\
ETH Zürich\\
Rämistrasse 101\\
8092 Zürich\\
Switzerland\\
E-mail: metalina.spohn@stat.math.ethz.ch}

\end{aug}



\end{frontmatter}

\section{Introduction}


Dealing with missing values is an integral part of modern statistical analysis. In particular, the assumed mechanism leading to the missing values is of great importance. Based on the work of \citet{Rubin1976}, there are three groups of missingness mechanisms usually considered: The values may be missing completely at random (MCAR), meaning the probability of a value being missing does not depend on the observed or unobserved data. In contrast, the probability of being missing could depend on observed values (missing at random, MAR) or on unobserved values (missing not at random, MNAR).

As stated in \citet{yuan2018}, ``a formal confirmation of the MCAR missing data mechanism is of great interest, simply because essentially all methods can still yield consistent estimates under MCAR even if the underlying population distribution is unknown''. While there is, at least for imputation, a number of approaches that can deal with a MAR missing data mechanism such as Multivariate Imputation by Chained Equations (mice) \citep{buuren_mice, Deng2016}, many commonly used methods explicitly rely on the validity of the MCAR assumption. Examples are the easy-to-use listwise-deletion and mean-imputation methods \citep{RubinLittlebook}. Consequently, the original paper on MCAR testing \citep{littletest} has been cited over $7600$ times according to google scholar. Recent papers (involving psychometric analysis) that test the MCAR assumption in order to justify listwise-deletion include
\cite{paperusingMCARtest4},
\cite{paperusingMCARtest2},
\cite{paperusingMCARtest1},  \cite{paperusingMCARtest3},  \cite{paperusingMCARtest5}, and \cite{paperusingMCARtest6}. As such, it is important to reliably test the MCAR assumption. 




The testing framework is of an ANOVA-type: when observing a dataset with missing values, there are $n$ observations and $G$ missingness patterns, $g=1,\ldots,G$. The observations belonging to the missingness pattern $g$ can be seen as a group, such that we observe $G$ groups of observations. The MCAR hypothesis now implies that the distribution of the observed data in all groups is the same, while under the alternative at least two differ. This is technically testing the observed at random (OAR) assumption defined in \citet{Rhoads2012}, see also the end of Section \ref{scoredefsec} for a discussion. This distinction can be avoided by assuming the missingness mechanism is MAR, which is what is usually implicitly done \citep{Li2015}.


The idea of testing the MCAR assumption traces back to \citet{littletest}. While some more refined versions of this testing idea were developed since then \citep{chen_Little, Kim2002, Jamshidian2010}, there has not been a lot of progress on distribution-free MCAR tests, able to detect general distributional differences between the missingness patterns. \citet{Li2015} recently made a step in that direction. Their test is completely nonparametric and shown to be consistent. Empirically it is shown to keep the level and to have a high power over a wide range of distributions. An application area where their proposed test struggles is for higher-dimensional data with little or no complete observations. Their testing paradigm is based on ``a reasonable amount'' of complete cases and all pairwise comparisons between the observed parts of two missingness pattern groups. This is problematic, since, as the dimension $p$ increases, the number of distinct patterns $G$ tends to grow quickly as well. The most extreme case occurs when $G = n$, that is, every observation forms a missingness pattern group on its own. Consequently, their test appears computationally prohibitively expensive for $p > 10$. Additionally, as the dimension increases, both the number of complete cases and the number of observations per pattern tends to decrease, both contributing to a reduction in power for the test in \citet{Li2015}.

In this paper, we try to circumvent these problems in a data-efficient way, by employing a one v.s. all-others approach and using \emph{random projections} in the variable space. Considering observations that are projected into a lower-dimensional space allows us to recover more complete cases. As realized by \citet{Li2015}, the problem of MCAR testing, as described above, is a problem of testing whether distributions across missingness patterns are different. The method presented here relies on some of the core ideas of \citet{michel2021proper} and \citet{Cal2020}, who do distributional testing using classifiers. We extend the ideas of \citet{Cal2020} to be usable for multiclass classification and use the projection idea of \citet{michel2021proper} to build a test that is usable and powerful even for high dimensions. Moreover, using a permutation approach, we are able to provably keep the nominal level $\alpha$ for all $n$. As outlined later, this is in contrast to other tests, for which the level might be kept only asymptotically, or is even unclear. The approach of random projections together with a permutation test also allows to extract more information than just a global hypothesis test. We make use of this to calculate individual $p$-values for each variable. Such a partial test for a variable addresses the null hypothesis that, once that variable is removed, the data is MCAR. Together with the test of overall MCAR, this might point towards the potential source of deviation from the null, that is, the variables causing an MCAR violation.



The paper is structured in the following way. Section \ref{problemformulationsection} introduces notation. Section \ref{scoredefsec} details the testing framework including the null and alternative hypotheses we consider. Section \ref{pracaspects} then showcases how to perform this test in practice and details the algorithm. Section \ref{empresults_1} shows some numerical comparisons for type-I error control and power. Section \ref{extension} explains the extension of partial $p$-values, while Section \ref{discuss} concludes. Appendix \ref{proofsection} contains the proofs of all results, while Appendix \ref{app:comptime} adds some additional details and shows computation times of the different tests.

\subsection{Contributions}

Our contributions can be summarized as follows: We develop the PKLM-test, an easy-to-use and powerful non-parametric test for MCAR, that is applicable even in high dimensions. We thereby extend the testing approach of \citet{Cal2020} to multiclass testing, which in connection with random projections in the variable space and the Random Forest classifier leads to a powerful test for both discrete and continuous types of data. To the best of your knowledge, no other test is as widely applicable and powerful. Moreover, we are able to formally prove the validity of our $p$-values for any sample size and number of groups $G$. As we demonstrate in our simulations, this is remarkable for the MCAR testing literature. It appears no other MCAR test has such a guarantee and many have inflated type-I errors, even in realistic cases, see e.g. the discussion in \citet{Jamshidian2010}.

As an extension, we can compute partial $p$-values corresponding to each variable, addressing the question of the source of violation of MCAR among the variables. We demonstrate the validity and power of our test on a wide range of simulated and real datasets in conjunction with different MAR mechanisms. Finally, we make our test available through the \textsf{R}-package \texttt{PKLMtest}, available on  \url{https://github.com/missValTeam/PKLMtest} and on CRAN. 




\subsection{Related Work}\label{relwork}

Previous advances for tests of MCAR were mostly addressed by \citet{littletest} (referred to as ``Little-test'') and extensions \citep{chen_Little, Kim2002} under the assumption of joint Gaussianity. To the best of our knowledge, the only distribution-free tests are developed in \citet{Jamshidian2010}, \citet{Li2015} and \citet{empiricallikelihoodapproach}. The first paper develops a test (referred to as ``JJ-test''), which is distribution-free but is only able to spot differences in the covariance matrices between the different patterns. As such, the simulation study in \citet{Li2015} shows that their test (referred to as ``Q-test''), which can detect any potential difference, has much more power than the JJ-test. Moreover, the JJ-test requires prior imputation of missing values, which appears undesirable. \citet{empiricallikelihoodapproach} develop a test that can be used to subsequently also consistently estimate certain estimators under MCAR. Their test requires a set of fully observed ``auxiliary'' variables that can be used to first test and then estimate properties of some variable of interest. As such their approach and goals are quite different from ours. 

Consequently, the test closest to ours is the fully non-parametric method in \citet{Li2015}. However, it is computationally costly or even infeasible to use their test with dimensions typically found in modern datasets ($p \gg 10$), as all pairwise comparisons between missingness patterns are calculated. While this could in principle be avoided by only checking a subset of pairs, we empirically show that, even if all pairwise comparisons are performed, our test has comparable or even higher power than theirs in their own simulation setting. This gap only increases with the number of dimensions or with a decrease in the fraction of fully observed cases.





We also address a major issue in the MCAR testing literature: none of the proposed methods has a finite sample guarantee of producing valid $p$-values and for some it can even be empirically checked that the produced $p$-value is not valid in certain settings. If $Z$ is a $p$-value generated from a statistical test, then it is valid if $\Prob(Z \leq \alpha) \leq \alpha$ under $H_0$ for all $\alpha \in [0,1]$, see e.g., \citet{lehmann2005testing}. Figure \ref{fig:ecdf} in Section \ref{empresults_1} shows some example of previous tests violating this validity of $p$-values. This issue might be surprising since the requirement of a valid $p$-value might be the most basic demand a statistical test needs to meet. For the Little-test, this is generally true under normality or asymptotically, that is if the number of observations is going to infinity, under some moment conditions and conditions on the group size. Despite this, Section \ref{empresults_1} shows that type error rates can strongly exceed the desired level even in samples of $500$ observations. The same holds for the JJ-test of \citet{Jamshidian2010} for which we sometimes observed a strong inflation of the level. As with the JJ-test, \citet{Li2015} also do not provide a formal guarantee that the level is kept. Though in our own simulation study, which is similar to theirs, we did not find any notable violation of the level for their test. 


To conduct our test, we adapt and partially extend the approaches of \citet{Cal2020} and \citet{michel2021proper}. The former develops a two-sample test using classification, an approach that has gained a lot of attention in recent years (see e.g., \citet{DPLBpublished} or \citet{hediger2020use} for a literature overview). We extend this approach to multiclass testing, to obtain a test statistic akin to \citet{Cal2020}, but using the out of bag (OOB) probability estimate of the Random Forest (RF) instead of the in-sample probability. This was already hinted in \citet{hediger2020use} to increase the power of the two-sample testing approach designed by \citet{Cal2020}. \citet{michel2021proper}, on the other hand, use random projections to increase the sample efficiency in the presence of missing values. This simple idea makes our test applicable and powerful, even in high dimensions, and even if the number of patterns $G$ is the same as the number of observations. It can also provide additional information together with the rejection decision, as we demonstrate in Section \ref{extension}. Finally, through an efficient permutation testing approach, we are able to formally guarantee that our test produces valid $p$-values for \emph{any} $n$ and any number of groups $G$. It appears that the PKLM-test is the first MCAR test with such a guarantee.  
Table~\ref{advantagetable} summarizes some of the properties of different tests. In particular, ``mixed data types'' refers to a possible combination of continuous data (such as income) and discrete data (such as gender), while ``power beyond differences in first and second moments'' means the test is able to detect differences between distributions, even if their means or variances are identical. Though this is difficult to show formally, it appears quite clear that the nonparametric nature of our approach allows for the detection of differences in distributions between patterns, even if the missingness groups all share the same mean or covariance matrix. As outlined in \citet{yuan2018} this is crucial for the detection of general MCAR deviations and is not the case, for instance, for the widely used Little-test. Appendix \ref{app:equalgroupmeansandvar} studies a simulated MAR example taken from \citet{yuan2018}, whereby observed means and variances are approximately the same across different groups. Tests such as the Little-test have no power in this example, yet with our approach, we reach a power of $1$.

\begin{table}[H]
    \resizebox{1\textwidth}{!}{%
\begin{tabular}{l|l|l|l|l}
& PKLM                       & Q                   & Little & JJ    \\ \hline
Computational Complexity                                        & $\mathcal{O}(p n \log(n))$  & $\mathcal{O}(n^2 p)$  &     $\mathcal{O}(n p^2) $       &     $\mathcal{O}(n (p^2 + \log(n)))$       \\ \hline
Can be used without       & Yes                           & No                                        &  No     &       Yes       \\
complete observations         &                            &                                         &        &                \\ \hline
Mixed data types possible                         & Yes                           & No                                        & No     & No             \\ \hline
 Does not require initial imputation & Yes & Yes & Yes & No\\ \hline
Power beyond differences                        & Yes                         & Yes                                        & No     & No             \\
 in first and second moments                          &                          &                                         &      &              \\ \hline
\end{tabular}}
\caption{Illustration of some of the properties of various tests. For details on the calculation of the computational complexities we refer to Appendix \ref{app:comptime}.
\label{advantagetable}}





\end{table}


\section{Notation}\label{problemformulationsection}

We assume an underlying probability space $(\Omega, \mathcal{F}, \Prob)$ on which all random elements are defined. Along the lines of \cite{josse} we introduce the following notation: let $\mathbf{X}^* \in \mathbb{R}^{n \times p}$ be a matrix of $n$ complete samples from a distribution $P^*$ on $\R^p$. We denote by $\mathbf{X}$ the corresponding incomplete dataset that is actually observed. Alongside $\mathbf{X}$ we observe the missingness matrix $\mathbf{M} \in \{0,1\}^{n\times p}$, of which an entry $m_{ij} \in \{0,1\}$ is $1$, if entry $x^*_{ij}$ is missing, and $0$, if it is observed. Each unique combination in $\{0,1\}^p$ in $\mathbf{M}$ is referred to as a missingness pattern and we assume that there are $G \leq n$ unique patterns in $\mathbf{M}$. As an example, for $p=2$, we might have the pattern $(1,0)$ (first value missing, second observed), $(0,1)$ (first value observed, second missing) or $(0,0)$ (both values are observed). We do not consider the completely missing pattern, in this case $(1,1)$. 


We assume that each row $x_{i}$ ($x^*_{i}$) of $\mathbf{X}$ ($\mathbf{X}^*$) is a realization of an i.i.d. copy of the random vector $X$ ($X^*$) with distribution $P$ ($P^*$). Similarly, $M$ is the random vector in $\{0,1\}^p$ encoding the missingness pattern of $X$. Furthermore we assume that $P$ ($P^*$) has a density $f$ ($f^*$) with respect to some dominating measure. For a random vector $X$ or an observation $x$ in $\R^p$ and subset $A \subseteq \{1,\ldots,p\}$, we denote as $X_A$ ($x_A$) the projection onto that subset of indices. For instance if $p=3$ and $A=\{1,2\}$, then $X_A=(X_1, X_2)$ ($x_A=(x_1, x_2)$). For any set $C \subseteq \{1,\ldots, p \}$, we denote by $\mathbf{X}_{\bullet C}$ the matrix of $n$ observations projected onto dimensions in $C$, so that $\mathbf{X}_{\bullet C}$ is of dimension $n \times |C|$. Similarly, for $R \subseteq \{1,\ldots,n \}$, $\mathbf{X}_{R \bullet}$ denotes the matrix of observations in set $R$, over all dimensions, so that the dimension of $\mathbf{X}_{R \bullet}$ is given by $|R| \times p$. We denote by $F_g$ (respectively $f_g$) the complete distribution (density) of the data in the $g^{th}$ missingness pattern group. A quick overview of the notation including the use of indices for the number of missingness patterns, dimensions, observations, projections and permutations is given in Table \ref{Notationtable}.

\begin{table}[!htbp]
    \centering
    \begin{center}
 \begin{tabular}{|c|c c |} 
 \hline
 notation &  partial   &    full   \\ [0.5ex] 
 \hline\hline
distribution & $P$  & $P^*$   \\ 
dataset & $\mathbf{X}$ & $\mathbf{X}^*$  \\
observation in $\R^p$ & $x_i$ & $x_i^*$   \\
random vector & $X$ & $X^*$  \\
density & $f$ & $f^*$  \\
\hline \hline
number of missingness patterns & $G$ &\\
number of dimensions &  $p$ &\\
number of observations &   $n$ & \\
number of projections &  $N$ & \\
number of permutations & $L$ & \\
 \hline
\end{tabular}
\end{center}
    \caption{\textbf{Notation}: Summary of the notation used throughout the paper, with (``partial'') and without (``full'') considering the missing values.}
    \label{Notationtable} 
\end{table}

\section{Testing Framework}\label{scoredefsec}

In this section, we formulate the specific null and alternative hypotheses for testing MCAR considered by the PKLM-test. Recalling the notation of Section \ref{problemformulationsection}, a missingness pattern is defined by a vector of length $p$, consisting of ones and zeros, indicating which of the $p$ variables are missing in the given pattern.
We divide the $n$ observations into $g \in \{1,\ldots, G\}$ unique groups, such that the observations of each group share the same missingness pattern. Each group $g \in \{1,\ldots, G\}$ contains $n_g$ observations such that $n_1 + \ldots + n_G = n$. Let $F_g$ denote the joint distribution of the $p$ variables in the missingness pattern group $g$, such that the $n_g$ observations of the group $g$ are i.i.d. draws from $F_g$. As stated in \cite{Li2015}, testing MCAR can be formulated by the hypothesis testing problem
\begin{align}\label{test1}
    &H_0: F^*_1=F^*_2= \ldots= F^*_G  \notag\\ 
    &\text{v.s.} \\ 
    &H_A: \exists \ i\neq j \in \{1,\ldots, G\} \ \text{s.t.} \ F^*_i \neq F^*_j. \notag
\end{align}


We want to emphasize the use of $F^*$ in the testing problem \eqref{test1}, indicating that these hypotheses involve distributions we cannot access. Thus, \eqref{test1} needs to be weakened. Borrowing the notation of \cite{Li2015}, for missingness pattern group $g$ we denote with $\boldsymbol o_g$ and $\boldsymbol m_g$ the subsets of $\{1,\ldots ,p\}$ indicating which variables are observed and which are missing, respectively. We denote the induced distributions by $F_{g,\boldsymbol o_g}$ and $F_{g,\boldsymbol m_g}$. For two groups $i$ and $j$, we denote by $\boldsymbol o_{ij} := \boldsymbol o_i  \cap \boldsymbol o_j $ the shared observed variables of both groups. As mentioned in \cite{Li2015}, it is not possible to test \eqref{test1} reliably, since the distribution $F_{i,\boldsymbol m_i}$ of the unobserved variables is inaccessible. Thus, \citet{Li2015} consider the following hypothesis testing problem
\begin{align}
    &H_0: F_{i, \boldsymbol o_{ij}} = F_{j, \boldsymbol o_{ij}} \text{   }  \forall i \neq j \in \{1,\ldots, G\} \notag\\ 
    &\text{v.s.} \label{test2}\\ 
    &H_A: \exists \ i\neq j \in \{1,\ldots, G\} \text{ with } \boldsymbol o_{ij} \neq \emptyset \text{ s.t. }  F_{i, \boldsymbol o_{ij}} \neq F_{j, \boldsymbol o_{ij}}.\notag
\end{align}

The null hypothesis $H_0$ of \eqref{test2} is implied by $H_0$ of \eqref{test1}, but not vice-versa. In other words, if we can reject the null hypothesis of \eqref{test2}, we can also reject the null hypothesis of \eqref{test1}. But if the null hypothesis of \eqref{test2} cannot be rejected, there could still be a distributional change for different groups in the unobserved parts, so that the null hypothesis of \eqref{test1} is not true. In this case, the missingness mechanism would be MNAR. Thus, using the terminology of \citet{Rhoads2012}, \eqref{test2} tests the ``observed at random'' (OAR) hypothesis instead of the MCAR hypothesis. The differentiation can be circumvented by assuming that the missingness mechanism is MAR, which is the approach usually taken, see \citet{Li2015}.

The comparison of all pairs of missingness groups in the hypothesis testing problem \eqref{test2} is problematic however, as laid out in the introduction. In the following, we circumvent this problem in a data-efficient way, considering a one v.s. all-others approach and employing \emph{random projections} in the variable space. Considering observations that are projected into a lower-dimensional space allows us to recover more complete cases. Let $\mathcal{A}$ be the set of all possible subsets of $\{1,\ldots,p \}$ with at most $p-1$ elements. For $A \in \mathcal{A}$ we define by $\mathcal{N}_{A}$ the indices in $1,\ldots, n$ of observations that are observed with respect to projection $A$, i.e., observations of which the projection onto $A$ is fully observed. These observations may belong to different missingness pattern groups $g \in \{1,\ldots , G\}$. As an example, $x = (\texttt{NA},1,\texttt{NA},2, 4)$ and $y = (\texttt{NA},\texttt{NA},\texttt{NA},1,3)$ are not complete and not in the same group, however if we project them to the dimensions $A = \{4,5\}$, $x_A$ and $y_A$ are complete in this lower-dimensional space.

Additionally, to circumvent the problem of many groups with only a few members, we assign new \textit{grouping or class labels} to all observations in $\mathcal{N}_A$. To do so, we consider the set of projections $\mathcal{B}(A^c)$, which is defined as the power set of $\{1,\ldots, p\}\setminus A$. The set $\mathcal{B}(A^c)$ is never empty since $|A| \leq p-1$. For a given projection $B \in \mathcal{B}(A^c)$, we project all observations with indx in $\mathcal{N}_A$ to $B$ and form new collapsed missingness pattern groups $G(A,B)$, where $G(A,B)$ is the set of labels corresponding to distinct missingness patterns among observations with index in $\mathcal{N}_A$ projected to $B$. This is solely done to determine the grouping or class labels of observations with index in $\mathcal{N}_A$. If two observations with index in $\mathcal{N}_A$ are in the same overall missingness pattern group $g \in \{1, \ldots, G\}$, they also end up in the same collapsed group. The other direction is not true, that is the number of collapsed groups $|G(A,B)|$ is at most as large as the initial number of distinct groups $G$ among the observations with index in $\mathcal{N}_A$. Considering again $x = (\texttt{NA},1,\texttt{NA},2, 4)$ and $y = (\texttt{NA},\texttt{NA},\texttt{NA},1,3)$, if $B = \{1,2\}$, then observations $x$ and $y$ are not in the same missingness pattern group. However, if $B = \{1,3\}$, we assign the same class label to $x$ and $y$. Thus, given the projection $A$, we obtain a set of fully observed observations $\mathbf{X}_{\mathcal{N}_{A}, A}=\mathbf{X}_{\mathcal{N}_{A}, A}^*$, and given the projection $B$ we assign to them the $|G(A,B)|$ different class labels. Figure \ref{fig:IllustrationofPKLMapproach} provides a schematic illustration of projections $A$ and $B$ on a more complicated example with four observations, each corresponding to a different pattern (i.e., $n=G=4$). According to $B= \{2\}$, the first observation in $\mathbf{X}_{\mathcal{N}_{A}, A}$ obtains one collapsed class label whereas the second and third observation obtain another, common label, resulting in $|G(A,B)|=2.$

We are now equipped to formulate our one v.s. all-others approach with the hypothesis testing problem 
\begin{align}
    &H_0: F_{g,A} = \sum_{j\in G(A,B) \setminus g} \omega_j^g F_{j,A} \notag\\
    &\forall g \in  G(A,B), \forall B \in \mathcal{B}(A^c), \forall A \in \mathcal{A}\notag\\ 
    &\text{v.s.} \label{test3}\\ 
    &H_A:  F_{g,A} \neq \sum_{j\in G(A,B)\setminus g} \omega_j^g F_{j,A}.\notag\\
    & \text{for one} \ g \in G(A,B), B \in \mathcal{B}(A^c), A \in \mathcal{A}, \notag
\end{align}
where $F_{g,A}$ is the joint distribution of the observations of class $g$ with index in $\mathcal{N}_A$ and the groups $j\in G(A,B)$ are jointly determined by $A$ and $B$. Thus, we compare the distribution of the observed part with respect to $A$ of one group $g$ with the mixture of the observed parts of the rest of the groups. The weights $\omega_j^g$ are non-negative, sum to $1$, and are proportional to the respective fraction of observations in class $j$. 


\begin{figure}
    \centering
    \includegraphics[width=8cm]{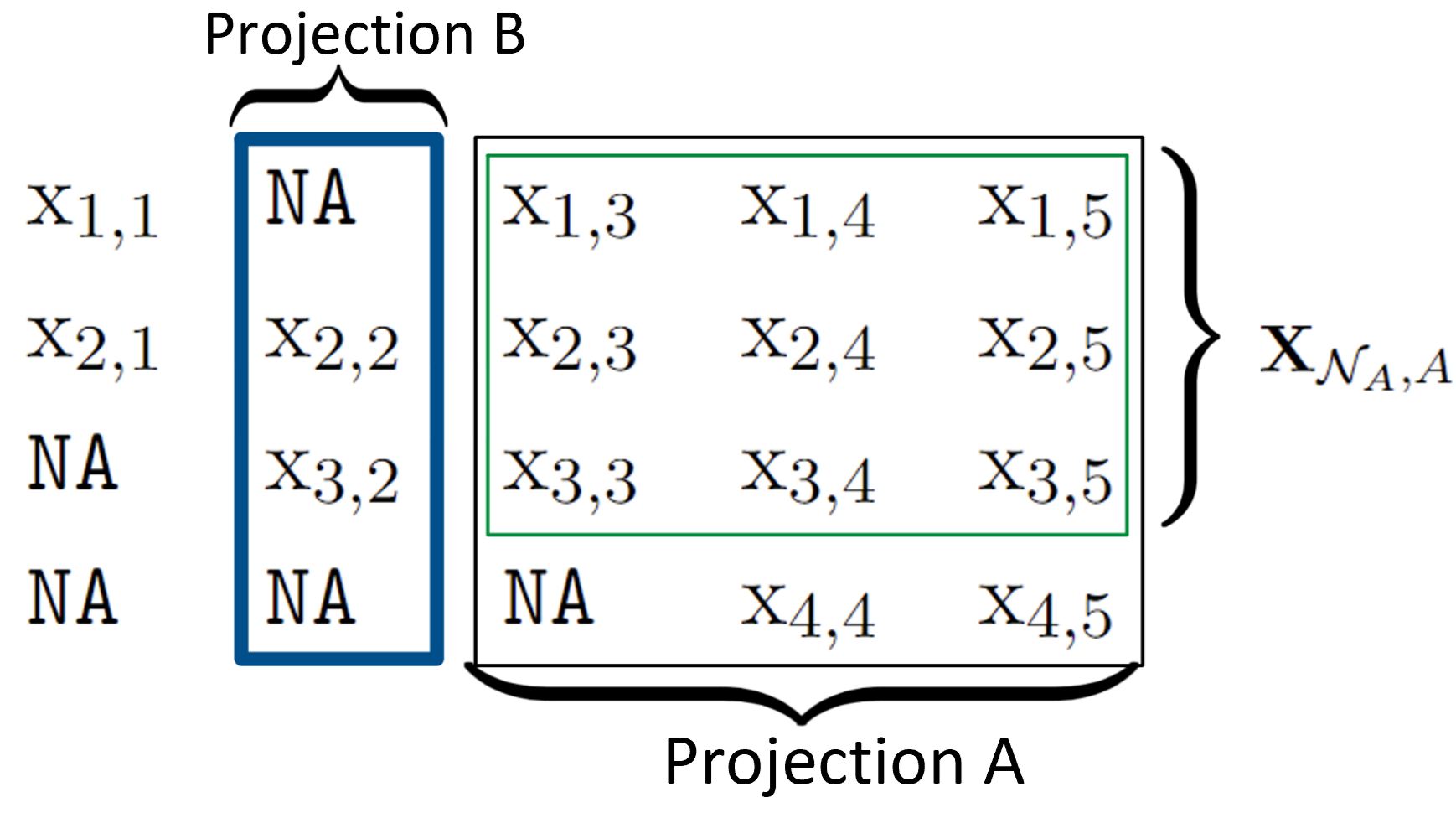}
    \caption{Illustration of the projections $A$ and $B$ in an example with $n=4$ and $p=5$. In a first step, a projection $A = \{3,4,5\} \subset \{1, \ldots, 5 \}$ is drawn. The fully observed points on $A$ form $\mathbf{X}_{\mathcal{N}_A,A}$, as indicated in green. In a second step, a projection $B= \{2\} \subset \{1, \ldots, 5 \} \setminus A$ is drawn, as indicated in blue. The patterns in projection $B$ then determine the labels assigned to the observations in $\mathbf{X}_{\mathcal{N}_A,A}$. In this case we obtain two different class labels: the first observation has one label, and the second and third observations share another common label.}
    \label{fig:IllustrationofPKLMapproach}
\end{figure}

\begin{example}
\textit{To give some intuition about the hypothesis testing problem \eqref{test3}, we relate it to the hypothesis testing problem \eqref{test2} with the help of the example of Figure \ref{fig:IllustrationofPKLMapproach}. In this example, each observation $i=1,\ldots, 4$ has a different pattern and can thus be seen as a draw from a distribution $F^*_i$. We first assume that the null hypothesis of \eqref{test3} holds and show, as an example, that this implies $F_{1,\boldsymbol o_{13}}=F_{3, \boldsymbol o_{13}}$. Since the null hypothesis of \eqref{test3} refers to all $A \in \A$, it also includes $A=\boldsymbol o_{13}=\{3,4,5\}$, which is what we consider in Figure \ref{fig:IllustrationofPKLMapproach}. While we are only interested in $F_{1,A}$ and $F_{3,A}$, taking $B=\{1,2\}$ the observations in $\mathbf{X}_{\mathcal{N}_A,A}$ come from the three distributions $F_{1,A}, F_{2,A}, F_{3,A}$. Due to \eqref{test3} it holds that \begin{align}\label{test3ex}
\begin{split}
    F_{1,A} &=  \omega_2^1 F_{2,A} + \omega_3^1 F_{3,A}, \\
    F_{2,A} &=  \omega_1^2 F_{1,A} + \omega_3^2 F_{3,A}, \\
    F_{3,A} &=  \omega_1^3 F_{1,A} + \omega_2^3 F_{2,A}.
    \end{split}
\end{align}
Some algebra shows that equation system \eqref{test3ex} is equivalent to $F_{1,A}=F_{2,A}=F_{3,A}$, which in particular means $F_{1,A}=F_{3,A}$, that we wanted to show. While we took $i=1$ and $j=3$ as an example matching Figure \ref{fig:IllustrationofPKLMapproach}, we cycle through all $A \in \A$ in \eqref{test3} and thus $A=\boldsymbol o_{ij}$ for all patterns $i,j$ eventually. We now assume that the null hypothesis of \eqref{test2} is true and consider again $A=\{3,4,5\}$ as an example. Since we only look at the fully observed observations in $\mathcal{N}_A$ in \eqref{test3}, i.e., leave out the fourth point, we again deal with the three distributions $F_{1,A}$, $F_{2,A}$, $F_{3,A}$. Moreover, by construction, $A \subset \boldsymbol o_{12}$ and $A \subset \boldsymbol o_{13}$ (even $A = \boldsymbol o_{13}$ in this case). Thus, $F_{1,\boldsymbol o_{12}}=F_{2,\boldsymbol o_{12}}$ and $F_{1,\boldsymbol o_{13}}=F_{3,\boldsymbol o_{13}}$, implied by the null hypothesis of \eqref{test2}, means that $F_{1,A}=F_{2,A}=F_{3,A}$, which implies \eqref{test3ex}. Again this might seem constructed, but since by definition, \eqref{test3} only considers the distributions $F_{i,A}$ and $F_{j,A}$ of fully observed points on $A$, it will always hold that $A \subset \boldsymbol o_{ij}$. }
\end{example}

We make note of an abuse of notation in \eqref{test3}, as the group $g$ in $F_{g,A}$ only corresponds to the same index of $F_g$ in \eqref{test2}, if $B=A^c$, as can be seen in the example of Figure \ref{fig:IllustrationofPKLMapproach}: If $B=A^c$, the three observations in $\mathbf{X}_{\mathcal{N}_A,A}$ are drawn from $F_{1,A}, F_{2,A}$ and $F_{3,A}$ respectively. However, if $B=\{2\}$, then observations two and three are now assumed to be drawn from a single distribution, which corresponds to a mixture of $F_{2,A}$ and $F_{3,A}$.  

In short, the null hypothesis of \eqref{test3} implies the null hypothesis of \eqref{test2} because for $A=\boldsymbol o_{ij}$, observations coming from $F_{i,A}$ and $F_{j,A}$ are contained in $\mathbf{X}_{\mathcal{N}_A,A}$. Vice-versa, the null hypothesis of \eqref{test2} implies the null hypothesis of \eqref{test3} because $A$ is nested in $\boldsymbol o_{ij}$ for all $F_i$ and $F_j$ considered on $A$. This actually sketches the proof of the following result:

\begin{restatable}{proposition}{propequivalence}
\label{equalitythmlabel}
Hypothesis testing problem \eqref{test3} is equivalent to \eqref{test2}.
\end{restatable}

Tackling hypothesis testing problem \eqref{test3} would be rather inefficient since we might test many times the same hypothesis when cycling through all $A \in \mathcal{A}$ and $B \in \mathcal{B}(A^c)$. However, the idea is that $A$ and $B$ will only be random draws from $\mathcal{A}$ and $\mathcal{B}(A^c)$. This is discussed in the next section.


\section{MCAR test Through Classification}\label{pracaspects}

In this section we introduce the classification-based statistic of our test and detail the implementation of our permutation approach, permuting the rows of the missingness matrix $\mathbf{M}$, to obtain a valid test.

\subsection{Test Statistic $U$}

Let us fix a projection $A \in \mathcal{A}$ and corresponding projection $B \in \mathcal{B}(A^c)$. 
We denote the induced collapsed class labels based on projections $A$ and $B$ by $Y^{(A,B)}$, by $X_A$ the projection of the random vector $X$ on $A$ and correspondingly by $x_A$ the projection on $A$ of observation $x$ in $\mathbf{X}_{\mathcal{N}_{A}, A}$. Furthermore, we define for each $g\in G(A,B)$ and $x$ in $\mathbf{X}_{\mathcal{N}_{A}, A}$ the following quantities:
\begin{align*}
   p^{(A,B)}_g(x) &:= P(Y^{(A,B)}=g \mid X_A=x_A), \\
   f^{(A,B)}_g(x) &:= P(x_A \mid Y^{(A,B)}=g),\\
   \pi^{(A,B)}_g & := P(Y^{(A,B)}=g).
\end{align*}
Let us fix $g\in G(A,B)$ as well. We reformulate the hypothesis testing problem (\ref{test3}):
\begin{align}
  H_{0,g}^{(A,B)}: &f^{(A,B)}_g = \frac{1}{1-\pi^{(A,B)}_g} \sum_{j\in\{1,\ldots, G(A,B)\}\setminus g} \pi^{(A,B)}_j f^{(A,B)}_j\notag\\
  &\text{v.s.} \label{test4}\\ 
H_{1,g}^{(A,B)}:  &f^{(A,B)}_g \neq \frac{1}{1-\pi^{(A,B)}_g} \sum_{j\in\{1,\ldots, G(A,B)\}\setminus g} \pi^{(A,B)}_j f^{(A,B)}_j.  \notag
\end{align}

Let $S_{f_g^{(A,B)}} \subset \mathcal{N}_A$ denote the indices of observations in $\mathbf{X}_{\mathcal{N}_{A}, A}$ that belong to class~$g$. For each missingness pattern~$g$, we now define the following statistic in analogy to \citet{Cal2020},
\begin{align}
    U^{(A,B)}_g := \frac{1}{|S_{f_g^{(A,B)}} |}\sum_{i\in S_{f_g^{(A,B)}}} \left(\log\frac{p^{(A,B)}_g(x_{i})}{1-p^{(A,B)}_g(x_{i})} - \log \frac{\pi^{(A,B)}_g}{1-\pi^{(A,B)}_g}\right).\label{Ustat}
\end{align} 

This statistic is motivated by the following claim:

\begin{restatable}{lemma}{densityratiothm}
\label{densityratiothmlabel}
The logarithm of the density ratio for testing (\ref{test4}) is given by $U^{(A,B)}_g$.
\end{restatable}



The main motivation for the form of this test-statistic is that one can use the same arguments as in \citet[Proposition 1]{Cal2020} to show that a test based on $U^{(A,B)}_g$ will have the highest power among all tests for \eqref{test4}, according to the Neyman-Pearson Lemma. In addition, the test statistic converges to the Kullback-Leibler Divergence between $f_g^{(A,B)}$ and the mixture of the other densities, motivating the name of our MCAR test. A high value of KL-Divergence indicates that the distributions of two samples deviate strongly from each other.


\begin{restatable}{lemma}{KLlemmalabel}
\label{KLlemmalabel}
$U_g^{(A,B)}$ converges in probability to the Kullback-Leibler Divergence between $f_g^{(A,B)}$ and the mixture of the other densities:
\begin{align*}
    U_g^{(A,B)} \rightarrow \E_{f_g}\left[  \log \frac{f^{(A,B)}_g(X)(1-\pi^{(A,B)}_g) }{ \sum_{j\in G(A,B)\setminus g} \pi^{(A,B)}_j f^{(A,B)}_j(X)}\right],
\end{align*} as $n_g$ and $\sum_{j\in \{1,\ldots,G\}\setminus g} n_j\rightarrow \infty$ and $n_g/n \rightarrow \pi^{(A,B)}_g \in (0,1)$.
\end{restatable}



Since the statistic $ U^{(A,B)}_g$ is evaluated only on cases $x \in S_{f_g^{(A,B)}}$, it holds that $ f^{(A,B)}_g(x)=  f^{*(A,B)}_g(x)$ and $p^{(A,B)}_g(x) = p^{*(A,B)}_g(x)$. This means that the projected complete and incomplete distributions coincide on the projected complete samples. Thus we are indeed asymptotically measuring the Kullback-Leibler Divergence between $f^{*(A,B)}_g$ and the mixture of the other densities.

Since there might be only very few observations for a single class $g$, we symmetrize the KL-Divergence. That is, we use the samples of all classes to evaluate the KL-Divergence and not only the samples of class $g$. Let $S_{f_{g^{c(A,B)}}} \subset \mathcal{N}_{A}$ denote the indices of observations in $\mathbf{X}_{\mathcal{N}_{A}, A}$ that belong to all other classes $G(A, B)\setminus g$. For each missingness pattern $g$, we will use, in the following, the difference between two of the above statistics, namely
\begin{align}\label{symm}
\begin{split}
    U^{(A,B)}_g -  U^{(A,B)}_{g^c} &= \frac{1}{| S_{f_g^{(A,B)}} |}\sum_{i\in S_{f_g^{(A,B)}}} \log\frac{p^{(A,B)}_g(x_{i})}{1-p^{(A,B)}_g(x_{i})} \\
    &- \frac{1}{| S_{f_{g^{c(A,B)}}} |}\sum_{i\in  S_{f_{g^{c(A,B)}}}} \log\frac{p^{(A,B)}_g(x_{i})}{1-p^{(A,B)}_g(x_{i})},
    \end{split}
\end{align} where the terms including the class probabilities $\pi_g^{(A,B)}$ cancel out. This difference converges to the symmetrized KL-Divergence between the mixture of $f^{(A,B)}_g$ and the remaining classes and is more sample efficient than only using $ U^{(A,B)}_g$. The test statistic for fixed $(A,B)$ is then given by
\begin{align*}
   U^{(A,B)} := \sum_{g=1}^{G(A,B)} (U^{(A,B)}_g - U^{(A,B)}_{g^c}),
\end{align*}
and the final test statistic is defined as
\begin{align}
   U := \E_{A \sim \kappa, B\sim \kappa(A^c)} [U^{(A,B)}].\label{testtatU}
\end{align}

\subsection{Practical Estimation of $U$}

We estimate $p_g^{(A,B)}$ with a multiclass-classifier, yielding $\hat{p}_g^{(A,B)}$. Plugging-in this quantity into \eqref{symm} yields $\hat{U}_g^{(A,B)}-\hat{U}_{g^c}^{(A,B)}$. We then estimate $U^{(A,B)}$ by
$$
 \hat{U}^{(A,B)} := \sum_{g=1}^{G(A,B)} (\hat{U}^{(A,B)}_g - \hat{U}^{(A,B)}_{g^c}).
$$

Finally, we estimate  $U$ by
\begin{equation}
    \hat{U}:=  \frac{1}{N}\sum_{i=1}^{N} \hat{U}^{(A_i,B_i)}, \label{finalstat}
\end{equation}
where $N$ is the number of draws of pairs of projections $(A_i,B_i)$, $i=1,\ldots,N$, with $A \in \mathcal{A}$ according to a distribution $\kappa$ and $B \in \mathcal{B}(A^c)$ according to a distribution $\kappa(A^c)$.

Our chosen multiclass classifier is Random Forest \citep{Breiman, Breiman2001}, more specifically, the probability forest of \citet{probabilityforests}. That is, for each of the $N$ projections, we fit a Random Forest with a specified number of trees, a parameter called $\texttt{num.trees.per.proj}$.
Thus, for each tree (or group of trees) a random subset of variables and labels is chosen based on which the test statistic is computed. In each tree, we set $\texttt{mtry}$ to the full dimension of the projection to not have an additional subsampling effect. This approach aligns naturally with the construction of Random Forest, as the overall approach might be seen as one aggregated Random Forest, which restricts the variables in each tree or group of trees to a random subset of variables. We finally use the OOB-samples for predicting $\hat{p}^{(A,B)}_g$.


The question remains how to sample the sets $(A_1, B_1), \ldots (A_N, B_N)$ at random. Our chosen approach is quite simple: we first randomly sample a number of dimensions $r_1$ by drawing uniformly from $\{1,\ldots,p-1 \}$. We then draw $r_1$ values without replacement from $\{1,\ldots,p \}$ to obtain $A$. Similarly, we randomly draw a value $r_2$ from $\{1,\ldots,p - r_1 \}$ and then draw $r_2$ values without replacement from $\{1,\ldots,p \}\setminus A$ to obtain $B$. We then consider $\mathbf{M}_{\mathcal{N}_{A}, B}$, i.e., all patterns for the fully observed observations in $A$ projected to $B$, and build the labels $Y^{(A,B)}$ based on the patterns in this matrix. This simple approach is used as a default, but one could also employ a more data-adaptive subsampling. In our algorithm, we might restrict the number of collapsed classes by selecting $B$ corresponding to $A$ accordingly. The parameter indicating the maximal number of collapsed classes allowed is given by \texttt{size.resp.set}. If set to $2$, we reduce the multi-class problem to a two-class problem. In Algorithm \ref{algo1} we provide the pseudo-code for the estimation of $\hat{U}^{(A,B)}$.

\begin{algorithm}[H]
\normalsize


\textbf{Inputs}:  incomplete dataset $\mathbf{X}$, missingness indicator $\mathbf{M}$, projections $A$, $B$ \\
\textbf{Result}: $\hat{U}^{(A,B)}$\\
\textbf{Hyper-parameters}: number of trees per projection \texttt{num.trees.per.proj}, standard parameters of the Probability Forests \texttt{size.resp.set}\;
 - Recover the complete cases $\mathcal{N}_{A}$ with respect to $A$\;
 - Generate the $G(A,B)$ collapsed class labels $Y^{(A,B)}$ from $\mathbf{M}_{\mathcal{N}_{A}, B}$\;
 
 - Fit a multi-class probability forest with $\texttt{num.trees.per.proj}$ trees and $\texttt{mtry}$ full\;
  
  \For{$g=1,\ldots, G(A,B)$}{
  - Estimate $\hat{p}_g^{(A,B)}$ with the fitted forest above using out-of-bag probabilities\;
 - Return the log-likelihood contribution $\hat{U}^{(A,B)}_g-\hat{U}^{(A,B)}_{g^c}$ for class $g$\;
 }
  - Average the log-likelihood ratio contributions $\hat{U}^{(A,B)}_g-\hat{U}^{(A,B)}_{g^c}$ from the $G(A,B)$ collapsed classes $g$ to get the statistic $\hat{U}^{(A,B)}$\;
 
 \caption{$\text{Uhat}(\mathbf{X}, \mathbf{M}, A, B)$}
 \label{algo1}
\end{algorithm}

To ensure that the level is kept by a test based on the statistic $\hat{U}$ for any choice of $\kappa$ and $\kappa(A^c)$, we use a permutation approach, as detailed next.

\subsection{Permutation Test}

To ensure the correct level, we follow a permutation approach. Informally speaking, the permutation approach works in this context if the testing procedure can be replicated in exactly the same way on the randomly permuted class labels. This is not completely trivial in this case, as the labels are defined in each projection via the missingness matrix $\mathbf{M}$. It can be shown numerically that permuting the labels at the level of the projection does not conserve the level, as this is blind to the correspondence between the projections across the permutations.

The key to the correct permutation approach is to permute the rows of $\mathbf{M}$. That is, for $L$ permutations $\sigma_{\ell}$, $\ell=1,\ldots,L$, we obtain $L$ matrices $\mathbf{M}_{\sigma_{1}}, \ldots, \mathbf{M}_{\sigma_{L}}$ with only the rows permuted. Then we proceed as above: We sample $A \sim \kappa$, $B \sim \kappa(A^c)$ and for each permutation of rows $\sigma_{\ell}$, $\ell=1,\ldots,L$, we calculate $U_{g, \sigma_{\ell}}^{(A,B)} - U_{g^c, \sigma_{\ell}}^{(A,B)}$ as in \eqref{symm}. Using $\hat{p}^{(A,B)}_g$ instead of $p^{(A,B)}_g$ this results in $\hat{U}_{g, \sigma_{\ell}}^{(A,B)}$ and in the statistic
\begin{align*}
    \hat{U}^{(A,B)}_{\sigma_{\ell}} := \sum_{g=1}^{G(A,B)} \hat{U}^{(A,B)}_{g, \sigma_{\ell}}-\hat{U}^{(A,B)}_{g^c, \sigma_{\ell}}.
\end{align*}
We note that we do not need to refit the forest for this permutation approach to work. Instead, we can directly use $\hat{p}^{(A,B)}_g$ from the original Random Forest that we fitted on the original $\mathbf{M}$.

Finally, we calculate the empirical distribution of the test-statistic under the null, by calculating for $\ell=1,\ldots, L$,
\begin{equation}\label{finalestimatepermuted}
    \hat{U}_{\sigma_{\ell}}:=  \frac{1}{N}\sum_{j=1}^{N} \hat{U}^{(A_j,B_j)}_{\sigma_{\ell}}.
\end{equation}
The $p$-value of the test is then obtained as usual by
\begin{align}\label{p-val}
    Z:=\frac{\sum_{\ell=1}^L \Ind\{ \hat{U}_{\sigma_{\ell}} \geq \hat{U} \} +1 }{L+1}.
\end{align}



Then it follows from standard theory on permutation tests that $Z$ is a valid $p$-value:

\begin{restatable}{proposition}{pvalprop}
\label{pvallabel}
Under $H_0$ in \eqref{test1}, and $Z$ as defined in \eqref{p-val}, it holds for all $z \in [0,1]$ that
\begin{equation}\label{validpval}
    \Prob(Z \leq z) \leq z.
\end{equation}
\end{restatable}




\begin{algorithm}[htb]

\normalsize
\textbf{Inputs}: incomplete dataset $\mathbf{X}$ \\
\textbf{Result}: $p$-value\\
\textbf{Hyper-parameters}: number of pairs of projections $N$, number of permutations $L$, number of trees per projection \texttt{num.trees.per.proj}, standard parameters of the Probability Forests, maximal number of collapsed classes \texttt{size.resp.set}\;
 - Randomly permute the rows of $\mathbf{M}$ $L$ times to obtain $\mathbf{M}_{\sigma_1}, \ldots, \mathbf{M}_{\sigma_L}$\;
 \For{$j=1,\ldots,N$}{
  - Sample a pair of projections $(A_j,B_j)$ hierarchically according to $A_j \sim \kappa$ and $B_j \sim \kappa(A_j)$\; 
 - Calculate $\hat{U}^{(A_j,B_j)}=\text{Uhat}(\mathbf{X}, \mathbf{M}, A_j, B_j)$\;
 
 \For{$\ell=1,\ldots,L$}{
 
 - Calculate $\hat{U}_{\sigma_{\ell}}^{(A_j,B_j)}=\text{Uhat}(\mathbf{X}, \mathbf{M}_{\sigma_{\ell}}, A_j, B_j)$\;
 }

 }
 - Average the statistics $\hat{U}^{(A_j,B_j)}$, $\hat{U}_{\sigma_{\ell}}^{(A_j,B_j)}$ over the couples of projections $(A_j,B_j)$ to get the final statistic $\hat{U}$, $\hat{U}_{\sigma_{\ell}}$, $\ell=1,\ldots,L$ \;
 - Obtain the $p$-value with \eqref{p-val}\;
 \caption{PKLMtest($\mathbf{X}$)}
 \label{Implementationdetails}
\end{algorithm}

Algorithm \ref{Implementationdetails} summarizes the testing procedure.



\section{Empirical Validation}\label{empresults_1}

In this section, we empirically showcase the power of our test in comparison to recent competitors on both simulated and real data. The simulation setting is set up along the lines of \citet{Jamshidian2010} and \citet{Li2015} with a common MAR mechanism. For the real datasets we also add a random MAR generation through the function \texttt{ampute} of the \textsf{R}-package \texttt{mice}, see e.g., \citet{ amputepaper}. 

As we did throughout the paper, we refer to our test as ``PKLM'', the test of \citet{Li2015} as ``Q'', the test of \citet{littletest} as ``Little'' and finally the one of \citet{Jamshidian2010} as ``JJ''. The Little-test is computed with the \textsf{R}-package \texttt{naniar} \citep{naniar}, while the JJ-test uses the code of the  \textsf{R}-package \texttt{MissMech} \citep{missmech}. Finally, the code for the Q-test was kindly provided to us by the authors.

\subsection{Simulated Data}\label{sec: simulations}
We vary the sample size $n$, the number of dimensions $p$, and the number of complete observations, which we denote by $r$. Cases $1-8$ describe the following different data distributions, similarly as in \cite{Li2015} and in \cite{Jamshidian2010}: Throughout, $I_p$ is a covariance matrix with diagonal elements $1$ and off-diagonal elements $0$ while $\Sigma$ is a covariance matrix with diagonal elements $1$ and off-diagonal elements $0.7$:
\begin{enumerate}
    \item A standard multivariate normal distribution with mean $0$ and covariance $I_p$,
    \item a correlated multivariate normal distribution with mean $0$ and covariance $\Sigma$,
    \item a multivariate $t$-distribution with mean $0$, covariance $I_p$ and degree of freedom $4$,
    \item a correlated multivariate $t$-distribution with mean $0$, covariance $\Sigma$ and degree of freedom $4$,
    \item a multivariate uniform distribution which has independent uniform$(0, 1)$ marginal distributions,
    \item a correlated multivariate uniform distribution obtained by multiplying $\Sigma^{1/2}$ to the multivariate uniform distribution in 5,
    \item a multivariate distribution obtained by generating 
    $W = Z + 0.1Z^3$, where $Z$ is from the standard multivariate normal distribution,
    \item a multivariate Weibull distribution which has independent Weibull marginal distribution, and each Weibull marginal distribution has scale parameter $1$ and shape parameter $2$.
\end{enumerate}

The above implements the fully observed $\mathbf{X}^*$. To compute the type-I error, we then simulate the MCAR mechanism where each value in the $p$ columns of the missingness matrix $\mathbf{M}$ has a probability of $1-r^{1/p}$ being one and is otherwise zero. To compute the power, we simulate the MAR mechanism following the description in \citet{Li2015}: We generate $\mathbf{M}$ such that the first column consists only of zeros so that the first variable is fully observed. Further, each value in the remaining $p-1$ columns has a probability of $1-r^{1/(p-1)}$ being one, while the rest is zero. This results, on average, in $r$ rows in $\mathbf{M}$ with only zeros, and thus in $r$ fully observed rows in $\mathbf{X}$. Next, we sort the rows of $\mathbf{M}$ into two groups, those that will be fully observed (complete group) and those that will have at least one missing value (missing group). So far, the generation is still MCAR. However now, for each row $i=1,\ldots,n$ we compare $\mathbf{X}^*_{i, 1}$ with the mean of $\mathbf{X}^*_{\bullet, 1}$, denoted by $\bar{X}_{1}$. If $\mathbf{X}^*_{i, 1} < \bar{X}_{1}$, the corresponding row $i$ is placed into the complete group with probability 1/6, and with probability 5/6 into the missing group. That is, with probability 1/6, the row $i$ is paired with a row in $\mathbf{M}$ from the complete group, and with probability 5/6, it is paired with a row from the missing group. Thus, in this case it is 5 times more likely that the row is placed in the missing group. On the other hand, if $\mathbf{X}^*_{i, 1} \geq \bar{X}_{1}$ the situation reverses, and row $i$ is 5 times more likely to be associated with a row in $\mathbf{M}$ from the complete group. Assigning the rows of $\mathbf{X}^*$ successively to the rows of $\mathbf{M}$ like this results in $\mathbf{X}$ with MAR missingness.

Each experiment was rerun $\texttt{nsim}=300$ times to compute type-I error and power. We used the following default hyperparameter setting for the computation of our PKLM-test: number of permutations $\texttt{nrep} = 30$, number of projections $\texttt{num.proj} = 100$, minimal node size in a tree $\texttt{min.node.size} = 10$, number of fitted trees per projection $\texttt{num.trees.per.proj} = 200$ and maximal number of collapsed classes allowed in a projection $\texttt{size.resp.set} = 2$. We note that the choice of these hyperparameters is intriguingly simple: besides $\texttt{size.resp.set}$, it holds that ``higher values are better''. Thus, as with RF in general, it is mostly a question of computational resources determining how large the values can be chosen. This is especially true for the number of trees for each forest, which should be relatively high in order to minimize additional randomness. We found $\texttt{num.trees.per.proj} = 200$ to be a good compromise between speed and accuracy. As the level is guaranteed for any number of permutations, and we desired a choice of hyperparameters that would work for $p=4$ as well as $p=40$, we chose the number of permutations low ($\texttt{nrep}=30$), but the number of projections relatively high ($\texttt{num.proj} = 100$). The only ``difficult'' parameter to set is $\texttt{size.resp.set}$, as there appears to be some loss in accuracy when the number of classes is larger than two. We thus found that $\texttt{size.resp.set}=2$, generating two classes, works well in a wide range of examples.

As mentioned throughout the paper, the Q-test could not be calculated for a large range of settings.\footnote{The largest number $p$ reported in the paper of \cite{Li2015} is $10$, while $r$ is at least $0.35$.} In particular, computation times were infeasible for the setting $p=10$ and $r=0.1$, and for any configuration with $p = 20$ or $p=40$. For the setting $n= 500$, $p=10$ and $r=0.1$ for instance, one test for case $2$ took around 20 minutes to finish, implying an approximate overall computation time of $500 \cdot 8 \cdot 2 \cdot 20 = 16000$ minutes or approximately 110 24-hour days. This despite the fact that the \textsf{R}-code of the Q-test we received was well implemented. In the upcoming Tables \ref{tab3} and \ref{tab4} of results we always used the nominal level of $\alpha = 0.05$. We boldfaced the results for each row in the tables in the following manner: Whenever the type $I$ error of a test is below or equal to $0.05$ and the test has the best power, it will be boldfaced. If this is true for more than one test, they are all boldfaced. Additionally, we boldfaced all the type-$I$ errors that are below or equal to the nominal level $\alpha = 0.05$ to indicate which tests holds the level on average in the given settings. 

In the simulation set-up of $n=200$ and $p=4$, the Q-test is very powerful, while keeping the nominal level. The PKLM-test is rarely the most powerful here, however the power of the PKLM-test is often relatively close to the best power. As an example, in case $2$ for $r=0.65$, the Q-test has a power of $1$ while the PKLM-test has a power $0.93$, with both keeping the nominal level $\alpha=0.05$. 


In the set-up of $n=500$ and $p=10$, the overall picture changes. The PKLM-test is in all but two of the $24$ cases the most powerful test, sometimes leaving the second-best test quite far behind. As an example, in case $3$ for $r=0.65$, the PKLM-test has a power of $0.85$ while the Q- and the Little-test exhibit a power of $0.26$ and $0.61$, respectively. While the Little- and the JJ-test often show inflated levels, this is never a problem for the valid PKLM-test. 

In the simulation set-up of $n=500$ and $p=20$, it appears as if the Little-test is a strong competitor. But this is only until one considers its type-I error. Though to a much lesser degree than the JJ-test, the type-I error is often heavily larger than the nominal level. Considering for instance case $4$, the power of the Little-test is even slightly less than its actual type-I error for $r=0.1$. In case $4$ with $r=0.35$, our test displays a power of $0.89$ and keeps the level, while the Little-test only has a power of $0.33$ despite having a grossly inflated type-I error. All of these problems are worsened for the JJ-test, which often displays an inflated type-I error in almost all cases and simulation set-ups. A similar story plays out in the case $r=0.65$.

Finally, in the simulation set-up of $n=1000$ and $p=40$, the power of our test is again much better than that of all other tests. Interestingly, the PKLM-test tends to have higher power when the components of the distribution are not independent, such as in the cases 2, 4, 6, and 8. For example, in case $1$ for $r=0.65$, PKLM has a power of 0.2, while for case $2$ it has a power of $0.95$. The main difference between these two cases is the strong positive correlation induced in case $2$. This pattern repeats: in all correlated examples and for both $r=0.65$ and $r=0.35$, the PKLM has a power nearing $1$, whereas in the independent versions, the power is closer to the type-I error. Thus, our test is able to use the dependencies in the data to its advantage, at least for $r=0.65$ and $r=0.35$, and can reach a very high power even for comparatively large $p$.

In summary, our test is very competitive even in small dimensions, where the Q-test is very powerful. It leaves behind all other tests by a wide margin as soon as one increases $p$. The Q-test remains strong in these situations as well, but becomes quickly infeasible as either $p$ increases or the fraction of complete cases $r$ decreases. Crucially, only the PKLM-test and the Q-test are able to consistently keep the nominal level over all experiments, with the Little- and JJ-test showing blatant inflation of the type-I error in many situations. This is the case despite the fact that simply checking the type-I error for a single level $\alpha$ ($0.05$ in this case) is far from sufficient to analyse the validity of a $p$-value. 
\begin{figure}[htb]
\centering
\includegraphics[width=11cm]{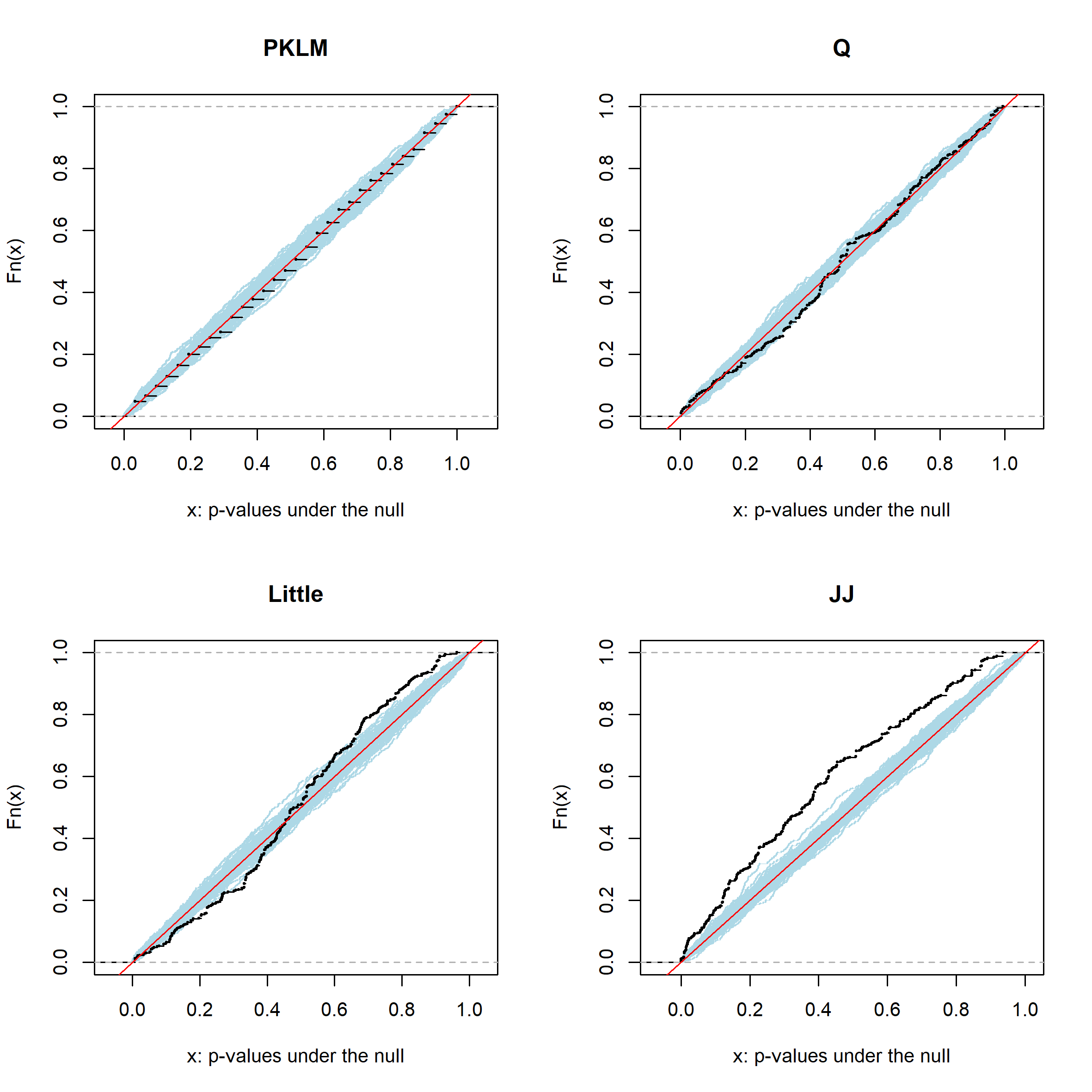}
\caption{Example plot of cumulative distribution function values of the $p$-values under the null (MCAR) of the four different tests. The simulation set up is $n=500$, $p=10$, $r=0.65$ in case $5$, with $500$ repetitions. The red line is the $x=y$ line, while the blue lines show $100$ ecdfs of $500$ simulated uniform random variables.}
\label{fig:ecdf}
\end{figure}

As an illustration, we randomly chose one of the above experiments in which the Little-test kept the nominal level, e.g., in the simulation set up $n=500$, $p=10$, $r=0.65$ in case $5$. In Figure \ref{fig:ecdf} we plot the empirical cumulative distribution functions (ecdf) of $500$ $p$-values under the null (MCAR) of the four different tests. The red line is the $x=y$ line. In blue we plotted $100$ ecdfs of a uniform$(0,1)$-distribution. As described in Equation \eqref{validpval}, a valid $p$-value has the property that the corresponding black ecdf values do not lie above the region defined by the blue lines. As Proposition \ref{pvallabel} predicts, this is clearly the case for the PKLM-test. That the $p$-values appear rather discrete stems from the fact that we chose a low number of permutations  ($\texttt{nrep}=30$). The Q-test is sometimes overshooting the red line, though this appears to mostly stem from estimation error. In general, it is remarkable how closely the ecdfs of $p$-values from both the Q- and PKLM-test resemble the ecdf of a uniform sample. The JJ-test appears to consistently have $P(Z \leq z) \geq z$. The Little-test finally appears to produce a valid $p$-value as long as only values $z < 0.5$ are considered. For $z\geq 0.5$, the the ecdf clearly violates the requirement of a valid $p$-value. If there is no theoretical guarantee, it is thus important to not just check the type-I error at $\alpha = 0.05$, but to instead consider other levels, e.g., $\alpha=0.1$.

\begin{table}
\centering
\begin{center}
\begin{adjustbox}{totalheight=\textheight}
\begin{tabular}{rrrl|rrrr|rrrr}
  \hline
  & & & &\multicolumn{4}{c}{Power} & \multicolumn{4}{c}{Type-I Error}\\
n & p & r & case & PKLM & Q & Little & JJ & PKLM & Q & Little & JJ \\ 
\hline
200 & 4 & 0.65 & 1 & 0.73 & \textbf{0.98} & 0.98 & 0.12 & \textbf{0.03} & \textbf{0.03} & 0.06 & \textbf{0.04} \\ 
   &  &  & 2 & \textbf{0.93} & 1.00 & 0.96 & 0.04 & \textbf{0.03} & 0.06 & 0.06 & \textbf{0.05} \\ 
   &  &  & 3 & 0.81 & \textbf{0.94} & 0.92 & 0.05 & \textbf{0.03} & \textbf{0.02} & \textbf{0.04} & 0.08 \\ 
   &  &  & 4 & 0.89 & \textbf{0.97} & 0.91 & 0.05 & \textbf{0.01} & \textbf{0.03} & \textbf{0.05} & \textbf{0.05} \\ 
   &  &  & 5 & 0.79 & \textbf{1.00} & \textbf{1.00} & 0.19 & \textbf{0.03} & \textbf{0.04} & \textbf{0.04} & 0.06 \\ 
   &  &  & 6 & 0.90 & \textbf{1.00} & 0.99 & 0.20 & \textbf{0.03} & \textbf{0.04} & \textbf{0.03} & 0.13 \\ 
   &  &  & 7 & \textbf{0.80} & 0.93 & 0.95 & 0.04 & \textbf{0.04} & 0.06 & 0.09 & 0.08 \\ 
   &  &  & 8 & 0.72 & \textbf{0.92} & 0.90 & 0.26 & \textbf{0.03} & \textbf{0.05} & \textbf{0.04} & 0.08 \\ 
   \hline
  200 & 4 & 0.35 & 1 & 0.79 & \textbf{0.98} & 0.97 & 0.04 & \textbf{0.03} & \textbf{0.04} & \textbf{0.04} & 0.13 \\ 
   &  &  & 2 & 0.87 & \textbf{0.98} & 0.97 & 0.08 & \textbf{0.03} & \textbf{0.03} & \textbf{0.03} & 0.08 \\ 
   &  &  & 3 & 0.82 & \textbf{0.97} & 0.90 & 0.16 & \textbf{0.03} & \textbf{0.03} & 0.06 & 0.12 \\ 
   &  &  & 4 & 0.87 & \textbf{0.99} & 0.92 & 0.10 & \textbf{0.03} & \textbf{0.02} & 0.08 & 0.11 \\ 
   &  &  & 5 & 0.79 & \textbf{0.99} & \textbf{0.99} & 0.10 & \textbf{0.04} & \textbf{0.05} & \textbf{0.05} & 0.08 \\ 
   &  &  & 6 & 0.80 & \textbf{1.00} & 0.97 & 0.12 & \textbf{0.03} & \textbf{0.04} & 0.06 & 0.11 \\ 
   &  &  & 7 & 0.79 & \textbf{0.98} & 0.92 & 0.09 & \textbf{0.03} & \textbf{0.05} & 0.07 & 0.06 \\ 
   &  &  & 8 & 0.83 & \textbf{0.99} & 0.99 & 0.10 & \textbf{0.05} & \textbf{0.05} & 0.06 & \textbf{0.05} \\ 
   \hline
  200 & 4 & 0.10 & 1 & 0.30 & \textbf{0.40} & 0.26 & 0.20 & 0.06 & \textbf{0.03} & \textbf{0.05} & 0.22 \\ 
   &  &  & 2 & \textbf{0.35} & 0.50 & 0.27 & 0.12 & \textbf{0.03} & 0.10 & \textbf{0.05} & 0.18 \\ 
   &  &  & 3 & 0.25 & \textbf{0.29} & 0.18 & 0.21 & \textbf{0.04} & \textbf{0.01} & \textbf{0.04} & 0.24 \\ 
   &  &  & 4 & 0.37 & \textbf{0.42} & 0.17 & 0.19 & \textbf{0.03} & \textbf{0.03} & \textbf{0.03} & 0.17 \\ 
   &  &  & 5 & 0.27 & \textbf{0.51} & 0.33 & 0.26 & \textbf{0.05} & \textbf{0.02} & \textbf{0.05} & 0.20 \\ 
   &  &  & 6 & 0.31 & \textbf{0.40} & 0.27 & 0.24 & \textbf{0.03} & \textbf{0.03} & \textbf{0.04} & 0.17 \\ 
   &  &  & 7 & 0.26 & \textbf{0.42} & 0.22 & 0.20 & \textbf{0.04} & \textbf{0.04} & 0.09 & 0.31 \\ 
   &  &  & 8 & 0.31 & \textbf{0.39} & 0.32 & 0.23 & \textbf{0.03} & \textbf{0.03} & \textbf{0.04} & 0.18 \\ 
   \hline
  500 & 10 & 0.65 & 1 & \textbf{0.93} & 0.89 & 0.88 & 0.09 & \textbf{0.05} & 0.06 & 0.06 & \textbf{0.05} \\ 
   &  &  & 2 & \textbf{0.99} & 1.00 & 0.84 & 0.08 & \textbf{0.02} & 0.06 & \textbf{0.05} & \textbf{0.05} \\ 
   &  &  & 3 & \textbf{0.85} & 0.26 & 0.61 & 0.12 & \textbf{0.02} & \textbf{0.05} & 0.18 & 0.10 \\ 
   &  &  & 4 & \textbf{0.99} & 0.96 & 0.60 & 0.10 & \textbf{0.04} & 0.06 & 0.19 & 0.12 \\ 
   &  &  & 5 & 0.89 & \textbf{0.98} & 0.96 & 0.16 & \textbf{0.04} & \textbf{0.05} & \textbf{0.03} & 0.10 \\ 
   &  &  & 6 & \textbf{0.99} & 1.00 & 0.91 & 0.15 & \textbf{0.04} & 0.07 & \textbf{0.02} & 0.13 \\ 
   &  &  & 7 & \textbf{0.90} & 0.61 & 0.68 & 0.09 & \textbf{0.02} & 0.07 & 0.12 & 0.07 \\ 
   &  &  & 8 & \textbf{0.79} & 0.76 & 0.76 & 0.18 & \textbf{0.03} & \textbf{0.04} & \textbf{0.05} & 0.09 \\ 
   \hline
  500 & 10 & 0.35 & 1 & \textbf{0.89} & 0.74 & 0.66 & 0.07 & \textbf{0.02} & \textbf{0.02} & \textbf{0.02} & 0.08 \\ 
   &  &  & 2 & \textbf{0.99} & 0.99 & 0.69 & 0.09 & \textbf{0.03} & 0.06 & \textbf{0.03} & 0.11 \\ 
   &  &  & 3 & \textbf{0.88} & 0.33 & 0.51 & 0.14 & \textbf{0.04} & \textbf{0.05} & 0.18 & 0.11 \\ 
   &  &  & 4 & \textbf{0.98} & 0.91 & 0.48 & 0.12 & \textbf{0.04} & 0.08 & 0.20 & 0.10 \\ 
   &  &  & 5 & \textbf{0.91} & 0.92 & 0.83 & 0.12 & \textbf{0.04} & 0.06 & \textbf{0.04} & 0.12 \\ 
   &  &  & 6 & \textbf{0.98} & 1.00 & 0.75 & 0.09 & \textbf{0.03} & 0.08 & \textbf{0.04} & 0.11 \\ 
   &  &  & 7 & \textbf{0.89} & 0.46 & 0.52 & 0.05 & \textbf{0.03} & \textbf{0.03} & 0.08 & 0.11 \\ 
   &  &  & 8 & \textbf{0.92} & 0.78 & 0.74 & 0.10 & \textbf{0.05} & 0.06 & 0.06 & 0.07 \\ 
   \hline
  500 & 10 & 0.10 & 1 & \textbf{0.31} & $-$ & 0.06 & 0.12 & \textbf{0.02} & $-$ & \textbf{0.03} & 0.10 \\ 
   &  &  & 2 & \textbf{0.45} & $-$ & 0.07 & 0.12 & \textbf{0.03} & $-$ & \textbf{0.03} & 0.07 \\ 
   &  &  & 3 & \textbf{0.34} & $-$ & 0.18 & 0.16 & \textbf{0.03} & $-$ & 0.19 & 0.14 \\ 
   &  &  & 4 & \textbf{0.45} & $-$ & 0.20 & 0.16 & \textbf{0.02} & $-$ & 0.22 & 0.11 \\ 
   &  &  & 5 & 0.33 & $-$ & \textbf{0.04} & 0.12 & 0.06 & $-$ & \textbf{0.02} & 0.14 \\ 
   &  &  & 6 & \textbf{0.45} &$-$  & 0.03 & 0.08 & \textbf{0.05} & $-$ & \textbf{0.01} & 0.12 \\ 
   &  &  & 7 & \textbf{0.34} & $-$ & 0.12 & 0.09 & \textbf{0.05} & $-$ & 0.09 & 0.15 \\ 
   &  &  & 8 & \textbf{0.34} & $-$ & 0.04 & 0.16 & \textbf{0.03} & $-$ & \textbf{0.05} & 0.13 \\ 
   \hline
\end{tabular}
 \end{adjustbox}
\end{center}
\caption{Simulated power and type-I error of PKLM, Q, Little and JJ for $n=200$, $p=4$ and $n=500$, $p=10$. We use $r=0.65$, $0.35$ and $0.1$. Cases $1-8$ describe different data distributions. The experiments were repeated $300$ times and the parameter setting for PKLM described above was used. }
\label{tab3}
\end{table}

\begin{table}
\centering
\begin{adjustbox}{totalheight=\textheight}
\begin{tabular}{rrrl|rrrr|rrrr}
  \hline
  & & & & \multicolumn{4}{c}{Power} & \multicolumn{4}{c}{Type-I Error} \\
n & p & r & case & PKLM & Q & Little & JJ & PKLM & Q & Little & JJ \\ 
  \hline
500 & 20 & 0.65 & 1 & \textbf{0.39} & $-$ & 0.36 & 0.06 &  \textbf{0.02} & $-$ &  \textbf{0.05} & 0.09 \\ 
   &  &  & 2 & \textbf{0.91} & $-$ & 0.48 & 0.08 & \textbf{0.03} & $-$ & \textbf{0.05} & 0.10 \\ 
   &  &  & 3 & \textbf{0.33} & $-$ & 0.49 & 0.20 & \textbf{0.03} & $-$ & 0.24 & 0.11 \\ 
   &  &  & 4 & \textbf{0.90} & $-$ & 0.40 & 0.14 & \textbf{0.04} & $-$ & 0.22 & 0.11 \\ 
   &  &  & 5 & 0.32 & $-$ & \textbf{0.64} & 0.14 & \textbf{0.04} & $-$ & \textbf{0.04} & 0.08 \\ 
   &  &  & 6 & \textbf{0.93} & $-$ & 0.39 & 0.25 & \textbf{0.04} & $-$ & \textbf{0.01} & 0.09 \\ 
   &  &  & 7 & \textbf{0.33} & $-$ & 0.37 & 0.07 & \textbf{0.03} & $-$ & 0.09 & 0.10 \\ 
   &  &  & 8 & \textbf{0.23} & $-$ & 0.25 & 0.14 & \textbf{0.04} & $-$ & 0.06 & \textbf{0.03} \\
   \hline
  500 & 20 & 0.35 & 1 & \textbf{0.45} & $-$ & 0.22 & 0.08 & \textbf{0.03} &  $-$& \textbf{0.04} & 0.09 \\ 
   &  &  & 2 & \textbf{0.90} & $-$ & 0.22 & 0.09 & \textbf{0.03} & $-$ & \textbf{0.04} & 0.08 \\ 
   &  &  & 3 & \textbf{0.43} & $-$ & 0.35 & 0.18 & \textbf{0.02} & $-$ & 0.34 & 0.12 \\ 
   &  &  & 4 & \textbf{0.89} & $-$ & 0.33 & 0.20 & \textbf{0.03} & $-$ & 0.31 & 0.15 \\ 
   &  &  & 5 & \textbf{0.46} & $-$ & 0.24 & 0.09 & \textbf{0.02} & $-$ & \textbf{0.02} & 0.12 \\ 
   &  &  & 6 & \textbf{0.91} & $-$ & 0.14 & 0.14 & \textbf{0.02} & $-$ & \textbf{0.03} & 0.10 \\ 
   &  &  & 7 & \textbf{0.41} & $-$ & 0.22 & 0.11 & \textbf{0.02} & $-$ & 0.11 & 0.10 \\ 
   &  &  & 8 & \textbf{0.52} & $-$ & 0.18 & 0.08 & \textbf{0.03} & $-$ & \textbf{0.04} & 0.07 \\
   \hline
  500 & 20 & 0.10 & 1 & \textbf{0.13} & $-$ & 0.00 & 0.14 & \textbf{0.03} & $-$ & \textbf{0.00} & 0.10 \\ 
   &  &  & 2 & \textbf{0.24} & $-$ & 0.01 & 0.14 & \textbf{0.04} & $-$ & \textbf{0.01} & 0.12 \\ 
   &  &  & 3 & 0.08 & $-$ & 0.21 & 0.16 & 0.06 & $-$ & 0.22 & 0.10 \\ 
   &  &  & 4 & \textbf{0.26} & $-$ & 0.27 & 0.08 & \textbf{0.04} & $-$ & 0.31 & 0.13 \\ 
   &  &  & 5 & \textbf{0.12} & $-$ & 0.00 & 0.10 & \textbf{0.03} & $-$ & \textbf{0.00} & 0.19 \\ 
   &  &  & 6 & \textbf{0.19} & $-$ & 0.00 & 0.11 & \textbf{0.05} & $-$ & \textbf{0.00} & 0.18 \\ 
   &  &  & 7 & \textbf{0.07} & $-$ & 0.08 & 0.12 & \textbf{0.04} & $-$ & 0.07 & 0.12 \\ 
   &  &  & 8 & \textbf{0.07} & $-$ & 0.02 & 0.11 & \textbf{0.04} & $-$ & \textbf{0.00} & 0.16 \\
   \hline
  1000 & 40 & 0.65 & 1 & \textbf{0.20} & $-$ & 0.00 & 0.09 & \textbf{0.05} & $-$ & \textbf{0.00} & 0.15 \\ 
   &  &  & 2 & \textbf{0.95} & $-$ & 0.00 & 0.12 & \textbf{0.03}  & $-$ & \textbf{0.00} & 0.14 \\ 
   &  &  & 3 & \textbf{0.23} &  $-$& 0.00 & 0.29 & \textbf{0.02} & $-$ & \textbf{0.00} & 0.17 \\ 
   &  &  & 4 & \textbf{0.94} & $-$ & 0.00 & 0.26 & \textbf{0.05} & $-$ & \textbf{0.00} & 0.17 \\ 
   &  &  & 5 & \textbf{0.16} & $-$ & 0.00 & 0.30 & \textbf{0.02} &  $-$& \textbf{0.00} & 0.19 \\ 
   &  &  & 6 & \textbf{0.97} & $-$ & 0.00 & 0.26 & \textbf{0.02} & $-$ & \textbf{0.00} & 0.19 \\ 
   &  &  & 7 & \textbf{0.23} & $-$ & 0.00 & 0.11 & \textbf{0.02} & $-$ & \textbf{0.00} & 0.10 \\ 
   &  &  & 8 & \textbf{0.13} &$-$  & 0.00 & 0.17 & \textbf{0.03} & $-$ & \textbf{0.00 }& 0.12 \\
   \hline
  1000 & 40 & 0.35 & 1 & \textbf{0.35} & $-$ & 0.00 & 0.12 & \textbf{0.02} & $-$ & \textbf{0.00} & 0.11 \\ 
   &  &  & 2 & \textbf{0.97} & $-$ & 0.00 & 0.13 & \textbf{0.05} & $-$ & \textbf{0.00} & 0.10 \\ 
   &  &  & 3 & \textbf{0.37} & $-$ & 0.00 & 0.30 & \textbf{0.03} & $-$ & \textbf{0.00} & 0.30 \\ 
   &  &  & 4 & \textbf{0.96} & $-$ & 0.00 & 0.33 & \textbf{0.04} & $-$ & \textbf{0.00} & 0.27 \\ 
   &  &  & 5 & \textbf{0.32} & $-$ & 0.00 & 0.14 & \textbf{0.04} & $-$ & \textbf{0.00} & 0.11 \\ 
   &  &  & 6 & \textbf{0.98} & $-$ & 0.00 & 0.16 & \textbf{0.03} & $-$ & \textbf{0.00} & 0.10 \\ 
   &  &  & 7 & \textbf{0.36} & $-$ & 0.00 & 0.11 & \textbf{0.02} & $-$ & \textbf{0.00} & 0.08 \\ 
   &  &  & 8 & \textbf{0.30} & $-$ & 0.00 & 0.16 & \textbf{0.02} & $-$ & \textbf{0.00} & 0.10 \\ 
   \hline
  1000 & 40 & 0.10 & 1 & \textbf{0.08} & $-$ & 0.00 & 0.15 & \textbf{0.02} & $-$ & \textbf{0.00} & 0.12 \\ 
   &  &  & 2 & \textbf{0.32} & $-$ & 0.00 & 0.12 & \textbf{0.02} & $-$ & \textbf{0.00} & 0.10 \\ 
   &  &  & 3 & \textbf{0.06} & $-$ & 0.00 & 0.13 &\textbf{ 0.05} & $-$ & \textbf{0.00} & 0.20 \\ 
   &  &  & 4 & \textbf{0.25} & $-$ & 0.00 & 0.25 & \textbf{0.03} & $-$ & \textbf{0.00} & 0.28 \\ 
   &  &  & 5 & \textbf{0.08} & $-$ & 0.00 & 0.11 & \textbf{0.03} & $-$ & \textbf{0.00} & 0.09 \\ 
   &  &  & 6 & \textbf{0.27} & $-$ & 0.00 & 0.09 & \textbf{0.04 }& $-$ & \textbf{0.00} & 0.11 \\ 
   &  &  & 7 & \textbf{0.07} & $-$ & 0.00 & 0.16 & \textbf{0.03} & $-$ & \textbf{0.00} & 0.13 \\ 
   &  &  & 8 & \textbf{0.07} & $-$ & 0.00 & 0.15 & \textbf{0.04} & $-$ & \textbf{0.00} & 0.08 \\ 
   \hline
\end{tabular} 
 \end{adjustbox}
\caption{Simulated power and type-I error of PKLM, Q, Little and JJ for $n=500$, $p=20$ and $n=1000$, $p=40$. We use $r=0.65$, $0.35$ and $0.1$. Cases $1-8$ describe different data distributions. The experiments were repeated $300$ times and the parameter setting for PKLM described above was used. }
\label{tab4}
\end{table}

\subsection{Real Data}

\begin{table}[h]
\centering
\resizebox{\textwidth}{!}{%
\begin{tabular}{lrr|rrrr|rrrr}
  \hline
   & && \multicolumn{4}{c}{Power} & \multicolumn{4}{c}{Type-I Error}   \\
dataset & n & p & PKLM & Q & Little & JJ & PKLM & Q & Little & JJ \\ 
  \hline
iris & 150 &   4 & 0.41 & \textbf{0.91} & 0.84 & 0.27 & \textbf{0.03} & \textbf{0.04} & \textbf{0.03} & 0.16 \\ 
  blood.transfusion & 748 &   4 & 0.48 & 0.97 & \textbf{1.00} &  \texttt{NA}  & \textbf{0.01} & 0.06 & \textbf{0.04} & \texttt{NA} \\ 
  airfoil & 1503 &   6 & \textbf{0.92} & 0.13 & 0.17 & 0.09 & \textbf{0.02} & \textbf{0.03} & 0.06 & 0.42 \\ 
  seeds & 210 &   7 & 0.64 & \textbf{0.74} & 0.57 & 0.24 & \textbf{0.05} & \textbf{0.02} & \textbf{0.02} & 0.10 \\ 
  yacht & 308 &   7 & 0.60 & 0.56 & \textbf{0.76} & 0.24 & \textbf{0.03} & 0.07 & \textbf{0.05} & 0.24 \\ 
  yeast & 1484 &   8 & \textbf{0.82} & 0.52 & 0.15 & 0.14 & \textbf{0.05} & 0.06 & 0.23 & 0.85 \\ 
  glass & 214 &   9 & 0.10 & 0.02 & \textbf{0.20} & 0.20 & \textbf{0.01} & \textbf{0.00} & \textbf{0.03} & 0.33 \\ 
  concrete.compression & 1030 &   9 & 0.64 & 0.48 & \textbf{0.81} & 0.47 & \textbf{0.04} & \textbf{0.04} & \textbf{0.05} & 0.41 \\ 
  wine.quality.red & 1599 &  11 & \textbf{0.81} & $-$ & 0.72 & 0.80 & \textbf{0.04} & $-$ & 0.15 & 0.52 \\ 
  wine.quality.white & 4898 &  11 & \textbf{0.98} & $-$ & 0.96 & 0.87 & \textbf{0.04} & $-$ & 0.10 & 0.79 \\ 
  planning.relax & 182 &  12 & \textbf{0.29} & $-$ & 0.20 & 0.14 & \textbf{0.00} & $-$ & \textbf{0.00} & \texttt{NA}  \\ 
  climate.model.crashes & 540 &  19 & \textbf{0.18} & $-$ & 0.22 & 0.47 & \textbf{0.00} &  $-$& \textbf{0.00} & \texttt{NA}  \\ 
  ionosphere & 351 &  32 & \textbf{0.45} & $-$ & 0.97 & 0.18 & \textbf{0.00} & $-$ & 0.06 & \texttt{NA} \\ 
   \hline
\end{tabular} %
}
\caption{Simulated power and level of PKLM, Q, Little and JJ for $13$ real datasets. We use $p_{miss}=0.3$. The experiments were repeated $300$ times and the parameter setting for PKLM described above was used. The \texttt{NA}s for some values of the JJ-test indicate that the test was not computable in any of the $300$ repetitions due to not enough observations in enough usable missingness groups.}\label{realresults}
\end{table}

We used $13$ real datasets with varying number of observations $n$ and dimensions $p$ for further empirical assessment of the PKLM-test and comparison to the other three tests. The datasets are available in the UCI machine learning repository\footnote{\url{https://archive.ics.uci.edu/ml/index.php}}. We preprocessed the data by cancelling factor variables, in order to be able to run all other three tests. However, we kept numerical variables with only few unique values. 

For the generation of the \texttt{NA}s, we use an overall probability of missingness of $p_{miss}=0.3$ (not to be confused with $r$ from the last subsection, denoting the number of complete cases).  We used a random MAR generation through the function \texttt{ampute} of the \textsf{R}-package \texttt{mice}. This function can randomly generate realistic MAR mechanisms, see e.g., \citet{ amputepaper}. Each experiment was run $\texttt{nsim}=300$ times to compute the type-I error and power. We used the following hyperparameter setting for the computation of our PKLM-test: number of permutations $\texttt{nrep} = 30$, number of projections $\texttt{num.proj} = 300$, minimal node size in a tree $\texttt{min.node.size} = 10$, number of fitted trees per projection $\texttt{num.trees.per.proj} = 200$ and maximal number of collapsed classes allowed in a projection $\texttt{size.resp.set} = 2$. The results are shown in Table \ref{realresults}. Our test is again very competitive with the best power in $7$ out of $13$ datasets, conditional on valid type-I errors. The Little-test shows also often good performance, though given the problematic level displayed in the previous section, this has to be considered with some care. 
The Q-test also has relatively high power in the situations where it can be calculated. However, due to computational time we only run the Q-test for $p \leq 10$. All in all, we see that the Q-test quickly gets infeasible for large $p$ and $n$ and the advantage of the PKLM-test strengthens with increasing $p$.

\section{Extension}\label{extension}

In addition to the ``global test'' of MCAR, we can study the effect of single variables: For any given variable $k=1,\ldots,p$, we can calculate 
\[
\hat{U}^{-k}= \frac{1}{|\mathcal{P}_{-k}|}\sum_{i \in \mathcal{P}_{-k}} \hat{U}^{(A_i,B_i)} ,
\]
where $\mathcal{P}_{-k}$ are all pairs of projections $(A_i, B_i)$ from the $N$ randomly chosen ones, with $B_i$ not containing variable $k$. We can use the analogous calculation based on the permuted missingness matrix $\mathbf{M}$
\[
\hat{U}^{-k}_{\sigma_{\ell}}= \frac{1}{|\mathcal{P}_{-k}|}\sum_{j \in \mathcal{P}_{-k}} \hat{U}^{(A_j,B_j)}_{\sigma_{\ell}},
\]
to obtain the $p$-value as in \eqref{p-val}. This ``partial'' $p$-value is valid and corresponds to the effect of removing the patterns induced by variable $k$. Indeed, assume the difference in the distribution of two patterns stems from a variable $j$ alone. If $j \in B$, a perfect classifier will be able to reliably differentiate the two, leading to a high value for $\hat{U}^{-k}$ relative to the permutation values. If $j$ is not forming the labels, we will not test these two classes against each other and thus not be able spot this difference. As such, we might expect to see a high $p$-value for $\hat{U}^{-j}$, when variable $j$ is removed, but a tendency to low $p$-values for $\hat{U}^{-k}$, $k \neq j$.

\begin{figure}[h]
\centering
\includegraphics[width=8cm]{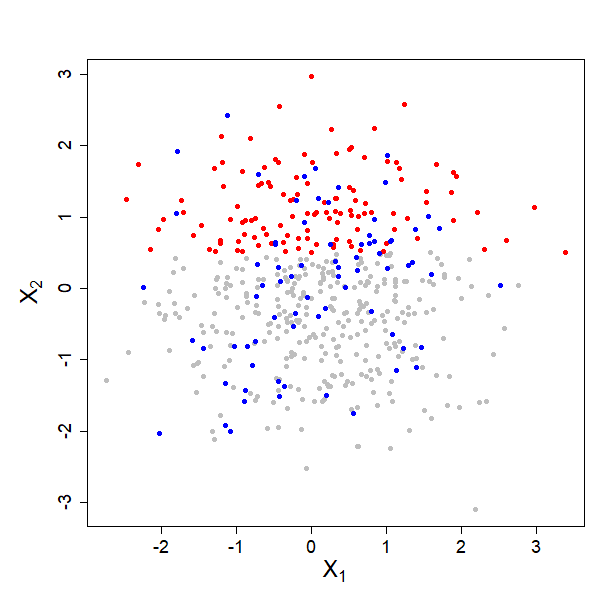}
\caption{$X_1$ and $X_2$ of the fully observed data in the simulated example of Section \ref{extension}. In red: Points with missing values in $X_1$, in blue: points with missing values in $X_2$. The blue points are randomly scattered, independently of the value of $X_1$, while in the red points, there is a visible trend towards having more missing values in $X_1$ for higher values of variable $X_2$.}
\label{fig:geo}
\end{figure}

We illustrate the usefulness of partial $p$-values with an example. Let $C_{-k} = \{1,\ldots,p\}\setminus \{k\}$. We assume $\mathbf{X}_{\bullet, C_{-k}}$ has a MCAR missingness structure, in particular, we simulate below the MCAR mechanism described in Section \ref{sec: simulations} with $r=0.65$. Let $k=1$ and assume that this first column of observations $\mathbf{X}_{\bullet,1}$ has missingness depending on the observed values of $\mathbf{X}_{\bullet, 2}$. For instance, each value is missing if the mean of the corresponding row $\mathbf{X}_{j,2}$ is larger than $0.5$. In this simple example $\mathbf{X}$ is MAR, but $\mathbf{X}_{\bullet, C_{-1}}$ is MCAR. We simulate this example, with $p=4$ and $n=500$, $\mathbf{X}_{i,\bullet}$ being independent standard Gaussian and the MAR/MCAR mechanism as described above. The first two fully observed components, $X_1$ and $X_2$, are shown in Figure \ref{fig:geo}. As before, we set \texttt{num.trees.per.proj}=$200$ and use $100$ projections. In this example, we are only able to spot any difference when $j=1$ is used to build the labels.

Our test reliably delivers small $p$-values ($ \leq 0.05$) for the three partial tests based on projections potentially including variable $1$, i.e., sets of projections $\mathcal{P}_{-2}$, $\mathcal{P}_{-3}$, and $\mathcal{P}_{-4}$ and a high $p$-value for the partial test based on $\mathcal{P}_{-1}$. Thus in this sense, the test detects that the main culprit of the MAR mechanism lies in the first variable.



\section{Concluding Remarks}\label{discuss}
In this paper we presented the powerful, flexible and easy-to-use PKLM-test for the MCAR assumption on the missingness mechanism of a dataset. We proved the validity of the $p$-value of the test and showed its power over a wide range of distributions. We also provided an extension allowing to do partial tests, that may shed light on the source of the violation of the MCAR assumption. Naturally, with some slight adaptations the test can be used as a general test of homogeneity of $G$ different groups in the sense that it tests whether $G$ different groups have the same distribution.

\newpage

\appendix
\setcounter{section}{0}
\section{Proofs} \label{proofsection}

\propequivalence*

\begin{proof}
We first show $H_0$ of \eqref{test2} implies $H_0$ of \eqref{test3}. Let $A, B$ be arbitrary. If they are such that there is only one label, there is nothing to test, so we may assume to have $|G(A,B)| \geq 2$ patterns in $\mathbf{X}_{\mathcal{N}_{A}, A}$. This means that $A \subset \boldsymbol{o}_{ij}$ for \emph{all} patterns $i,j \in G(A,B)$. This simply follows because, by construction, each of the $|G(A,B)|$ patterns in $\mathbf{X}_{\mathcal{N}_{A}, A}$ has the elements in $A$ fully observed. But since by assumption for all $i,j \in \{1,\ldots, G \}$, $F_{i,\boldsymbol{o}_{ij}}=F_{j,\boldsymbol{o}_{ij}}$ and $A\subset \boldsymbol{o}_{ij}$, this immediately implies that $F_{i,A}=F_{j,A}$ for all $i,j \in \{1,\ldots, G \}$ and thus $F_{g,A} = \sum_{j\in G(A,B)\setminus g} \omega_j^g F_{j,A}  $. Since $A,B$ were arbitrary, one direction follows.

We now show that $H_0$ of \eqref{test3} implies $H_0$ of \eqref{test2}. The proof is based on the following claim: Consider $G$ arbitrary distribution functions $F_1, \ldots, F_G$ and weights $(\omega^{g}_j)_{j=1}^{G-1}$, $j=1,\ldots, G$ such that $\sum_{j=1}^{G-1} \omega^{g}_j = 1$ for all $j$. Then 
\begin{align}\label{H001}
     F_{g} = \sum_{j\in\{1,\ldots, G\}\setminus g} \omega_j^g F_{j}, \text{   } \forall g \in  \{1,\ldots, G\}  \implies   F_{i}= F_{j}, \text{   }  \forall i \neq j \in \{1,\ldots, G\}.
\end{align}
We prove the implication by induction: Consider first $G=3$. Assuming the LHS of \eqref{H001} and plugging the equation for $F_{2}$ into the equation for $F_{1}$, we obtain:
\begin{align*}
    F_{1}&=w_2^1 w_1^2 F_{1}  + w_2^1 w_3^2 F_{3} + w_3^1 F_{3} \\
    &=w_2^1 w_1^2 F_{1}  + (w_2^1 w_3^2 + w_3^1) F_{3},
\end{align*}
which implies $(1-w_1^2w_2^1)F_{1}= (w_2^1w_3^2+w_3^1)F_{3}$. Since 
\begin{align*}
    &1 = w_2^1 + w_3^1
    = w_2^1(w_3^2+w_1^2) +w_3^1
    = w_2^1w_3^2 + w_2^1w_1^2 + w_3^1,
\end{align*}
we have the equality $(1-w_1^2w_2^1)=(w_2^1w_3^2+w_3^1)$ and thus $F_{1}=F_{3}$. Plugging this back into the equivalent equation for $F_{2}$, we obtain $F_{1}=F_{2}=F_{3}$. Now assume \eqref{H001} is true for $G$ distributions $F_{1}, \ldots, F_{G}$ and we now would like to prove it for $G+1$. Assume wlog that the weight of $F_{2}$ in the equation of $F_{1}$ is nonzero (there will always be at least one such distribution $F_{2}, \ldots, F_{G}$). Using the same trick as above, we may plug say the equation for $F_{2}$ into $F_{1}$, thereby reducing the number of equations/distributions to $G$. By the induction assumption this implies that $F_{1}=F_{3}=\ldots=F_{G}$. But immediately this also implies that $F_{2}=F_{1}$ and implies \eqref{H001}. With this result we can now proof the that $H_0$ of \eqref{test3} implies $H_0$ of \eqref{test2}. 

Take two arbitrary groups $i,j$ and $A=\boldsymbol{o}_{ij}$ and take $B=A^c$. To ease notation we just wlog take $i=1$ and $j=2$. Then $A=\boldsymbol{o}_{12}$ contains the dimensions for which patterns $1$ and $2$ have fully observed values. Thus, observations in $\mathbf{X}_{\mathcal{N}_{A}, A}$ contain draws from $F_{1,\boldsymbol{o}_{12}}$ and $F_{2,\boldsymbol{o}_{12}}$. Since by assumption
\begin{align}
     H_0: F_{g,A} = \sum_{j\in G(A,B)\setminus g}& \omega_j^g F_{j,A}  \text{   },  \forall g \in G(A,B),
\end{align}
it follows by \eqref{H001}, that $F_{i,A}=F_{j,A}$ for all $i,j \in G(A,B)$ and thus in particular, $F_{1,A}=F_{2,A}$. Since we will have $A=\boldsymbol{o}_{ij}$ for all groups $i\neq j$, $H_0$ of \eqref{test2} holds.
\end{proof}

\densityratiothm*
\begin{proof}
Based on the definitions of $  p^{(A,B)}_g(x)$, $f^{(A,B)}_g(x)$ and $\pi^{(A,B)}_g$ we obtain by Bayes Rule,
\begin{align} \label{bayes}
    p^{(A,B)}_g(x) = \frac{f^{(A,B)}_g(x)\pi^{(A,B)}_g}{\sum_{j\in G(A,B)} \pi^{(A,B)}_j f^{(A,B)}_j(x)},
\end{align}
assuming the existence of densities $f_g$ of distributions $F_g$ for each $g\in G(A, B)$.
Following the same steps as in \cite{Cal2020}, we get that the logarithm of the (joint) density ratio for testing $H_0$ vs $H_1$ of (\ref{test4}), given by
\begin{align}
    \log \frac{  f^{(A,B)}_g(x) (1-\pi^{(A,B)}_g)}{\sum_{j\in G(A,B)\setminus g} \pi^{(A,B)}_j f^{(A,B)}_j(x) }.\label{densityratio}
\end{align}
We reformulate the fraction in \eqref{densityratio} in terms of $p^{(A,B)}_g$, starting from (\ref{bayes}):
\begin{align*}
    p^{(A,B)}_g(x) \sum_{j\in G(A,B)\setminus g} \pi^{(A,B)}_j f^{(A,B)}_j(x) &= (\pi^{(A,B)}_g - p^{(A,B)}_g(x)\pi^{(A,B)}_g)f^{(A,B)}_g(x)\\
    &=\pi^{(A,B)}_g( 1  - p^{(A,B)}_g(x))f^{(A,B)}_g(x).
\end{align*} Thus, the inside of the logarithm of (\ref{densityratio}) is given by the following function of $p^{(A,B)}_g$:
\begin{align*}
     \frac{f^{(A,B)}_g(x)(1-\pi^{(A,B)}_g) }{ \sum_{j\in G(A,B)\setminus g} \pi^{(A,B)}_j f^{(A,B)}_j(x)   } &= 
     \frac{ 1-\pi^{(A,B)}_g}{\pi^{(A,B)}_g}\frac{p^{(A,B)}_g(x) }{1-p^{(A,B)}_g(x) }.
\end{align*}

\end{proof}

\KLlemmalabel*

\begin{proof}
From the proof of Lemma \ref{densityratiothmlabel}, we know that $U^{(A,B)}_g$ can be rewritten as
\begin{align}
    U^{(A,B)}_g &:= \frac{1}{|S_{f_g^{(A,B)}} |}\sum_{i\in S_{f_g^{(A,B)}}} \left(\log\frac{p^{(A,B)}_g(x_{i})}{1-p^{(A,B)}_g(x_{i})} - \log \frac{\pi^{(A,B)}_g}{1-\pi^{(A,B)}_g}\right) \notag\\
    &=\frac{1}{n_g}\sum_{i\in S_{f_g^{(A,B)}}} \log  \frac{f^{(A,B)}_g(x_i)(1-\pi^{(A,B)}_g) }{ \sum_{j\in G(A,B)\setminus g} \pi^{(A,B)}_j f^{(A,B)}_j(x_i)   }.\label{Ustat2}
\end{align} 
Since $n_g/n \rightarrow \pi^{(A,B)}_g \in (0,1)$ and the $x_i$ are i.i.d., the result follows from the law of large numbers.
\end{proof}

\pvalprop*

\begin{proof}

Let $\mathbf{A}=(A_1,\ldots, A_{N})$ and $\mathbf{B}=(B_1,\ldots, B_{N})$ be two sets of $N$ projections. Let $G_1, \ldots G_{L^*}$ be all possible permutations of the rows of the missingness matrix $\mathbf{M}$, such that
\[
G_{\ell}(\mathbf{X}^*, \mathbf{M}, \mathbf{A}, \mathbf{B})= (\mathbf{X}^*, \mathbf{M}_{\sigma_{\ell}}, \mathbf{A}, \mathbf{B}),
\]
for $\ell=1,\ldots, L^*$. Note that, since we are only considering fully observed observations for all projections in $\mathbf{A}$, $\hat{U}$, a function of $(\mathbf{X}, \mathbf{M}, \mathbf{A}, \mathbf{B})$, is indeed a function of $(\mathbf{X}^*, \mathbf{M}, \mathbf{A}, \mathbf{B})$, while $\hat{U}_{\sigma_{\ell}}$ is a function of $G_{\ell}(\mathbf{X}^*, \mathbf{M}, \mathbf{A}, \mathbf{B})$. It also holds that under the null, that is under MCAR, that
\begin{equation}\label{Equaldist}
    (\mathbf{X}^*, \mathbf{M}, \mathbf{A}, \mathbf{B}) \stackrel{D}{=} (\mathbf{X}^*, \mathbf{M}_{\sigma_{\ell}}, \mathbf{A}, \mathbf{B}) = G_{\ell}(\mathbf{X}^*, \mathbf{M}, \mathbf{A}, \mathbf{B})\ \ \forall \ell = 1,\ldots L^*.
\end{equation}
This is true because, under MCAR, $\mathbf{M}$ and $\mathbf{X}^*$ are independent. Since by the i.i.d. assumption also $\mathbf{M}_{\sigma_{\ell}} \stackrel{D}{=} \mathbf{M}$ for all $\ell= 1,\ldots, L^*$ and since $\mathbf{A}$, $\mathbf{B}$ are also independent of $\mathbf{M}$, \eqref{Equaldist} follows.
As outlined for example in \citet{Hemerik2018}, this implies that under $H_0$,
\[
\Prob(Z \leq z \mid \mathbf{A}, \mathbf{B}) \leq z.
\]
Integrating over $(\mathbf{A}, \mathbf{B})$, results in \eqref{validpval}.
\end{proof}

\section{Additional Details and Computation Times} \label{app:comptime}

Here we provide more implementation details, discuss the complexity calculations in Table \ref{advantagetable} and show computation times of the different tests in the experiments.

\noindent    
\textbf{Numerical truncation.} In order to avoid numerical issues when calculating the density ratio with Expression (\ref{Ustat}) or the $\log$ thereof, if we get predicted probabilities $\hat p_{A}$ close to $0$ or $1$, we apply the following truncation function to $\hat p_{A}$:
    \begin{equation*}
        p(x) = \min(\max(x, 10^{-9}), 1-10^{-9}).
    \end{equation*}

\noindent 
\textbf{Hyperparameter Selection.} 
Generally speaking, it holds that ``the more the better'', certainly for the parameters $N$, $L$ and \texttt{num.trees.per.proj}. As such, the choice of those three parameters depends mostly on the computational power available to the user. For \texttt{size.resp.set}, this is not quite as clear, though we found a value of two to work well in most situations.

\noindent    
\textbf{PLKM Test.} We first consider the complexity of one Random Forest, which is in this case
\[
\texttt{num.tree} \cdot p n\log(n).
\]
Note that this includes the calculation of $\hat{p}$ on the test sample through the OOB-error. In total we do this \texttt{num.proj} times. However, we consider \texttt{num.tree} and \texttt{num.proj} independent of $n$ and $p$ and thus treating it as constant. In this case we end up with $p n\log(n)$. Finally we need to calculate the statistics $U$ and repeat this number of calculations a fixed number of times. This would add a factor $Bn$, where again we assume that $B$ does not grow with $n$ and $p$. As this is neglible compared to $p n\log(n)$, the complexity is given as $\mathcal{O}(p n\log(n))$.
\noindent

\noindent    
\textbf{Q-test.} The Q-test compares all groups leading to a complexity of $G^2$ to compare each group with any other. Additionally, the statistic used is an MMD type, so the complexity is $(n_1+n_2)^2$, where $n_1$, $n_2$ are the respective group sizes. The group size can be at worst $n/G$, which together results in $\mathcal{O}(n^2)$. The bootstrap on the other hand can also be ignored, as it simply results in a constant factor multiplied to $n^2$.
\noindent

\noindent    
\textbf{JJ and Little-test.} Both JJ- and Little-test rely on covariance estimation which scales as $n p^2$. This gives the $\mathcal{O}(n p^2)$ complexity for the Little-test. For the JJ-test one also needs an ordering operation to obtain the test statistics, with complexity $n \log(n)$, which results in overall complexity $\mathcal{O}(n (p^2 + \log(n)))$.
\noindent

As mentioned above, Table \ref{advantagetable} just shows how the complexity scales in $n$ and $p$ and, in case of our test, treats the number of projections as constants. One might argue that the number of projections should be a function of $p$ as well. Similarly, for ``small'' $p$ and small number of groups $G$, the Q-test can be faster than ours. Still the complexities provide a good illustration of how quickly the Q-test can become infeasible, when the number of groups (often a function of $p$) and/or the number of observation increases.

\section{Example of \lowercase{\citet{yuan2018}}} \label{app:equalgroupmeansandvar}

\citet{yuan2018} study settings where group means and variances are approximately equal across missingness patterns, such that MCAR tests based on differences in means and variances, such as the Little-test, have no power. We study one such example here: Let $p=2$ and $(Z_1,Z_2)$ be jointly multivariate normal with correlation zero and let $X_1=Z_1$ and 
\[
X_2 = 0.5 Z_1 + (1-0.25)^{1/2}Z_2.
\]
We set $X_2$ to $\texttt{NA}$ if
\[
X_1 \in (-\infty, -1.932] \cup (-0.314, 0.314] \cup (1.932, \infty).
\]
This corresponds to around $30\%$ missing values. Figure \ref{fig:simulationillustration} displays a histogram, plotting all observations of $X_1$ with $X_2$ missing for a simulation of $n=10'000$. This corresponds to the MAR example used in \citet[Section 3]{yuan2018} and we refer to their paper for more details.

We simulate the above distribution for $n=1000$ and run our PKLM-test with the same parameters as described in Section \ref{sec: simulations}. Though the deviation from MCAR cannot be detected through the first two moments in this example, our test reliably reaches a power of 1.

\begin{figure}
    \centering
    \includegraphics[width=8cm]{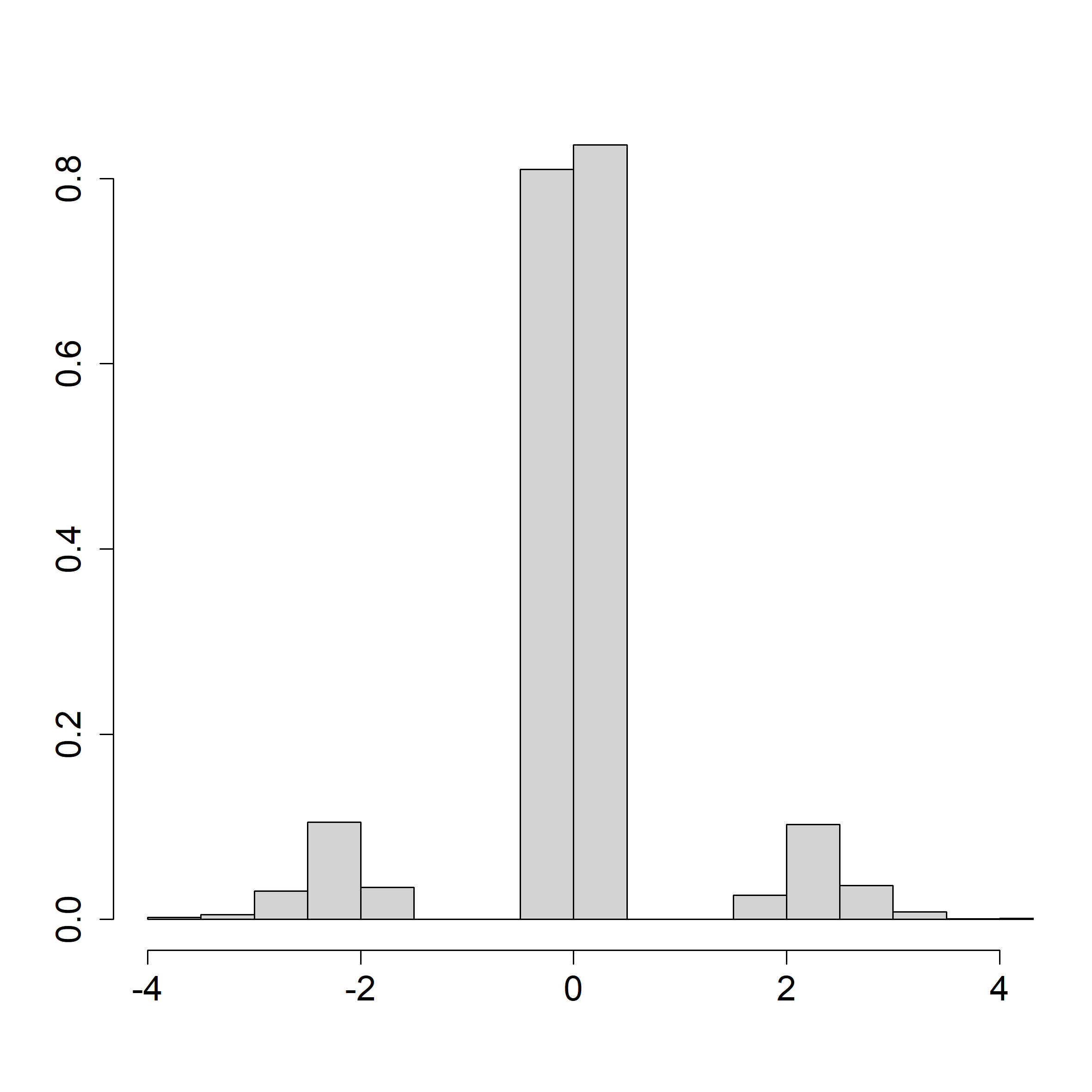}
    \caption{Histogram with relative frequencies of $X_1$ if the corresponding $X_2$ is $\texttt{NA}$.}
    \label{fig:simulationillustration}
\end{figure}

\vspace{\fill}\pagebreak




\bibliographystyle{apalike}
\bibliography{biblio}

\end{document}